\documentclass{theoretics}

\usepackage{amsfonts,amsmath,amsthm,amssymb,dsfont}

\usepackage{framed}
\usepackage[capitalise,noabbrev,nameinlink]{cleveref}
\usepackage{float,wrapfig}
\floatstyle{boxed} 
\restylefloat{figure}
\usepackage{mathrsfs}
\usepackage{mathdots}

\usepackage{xfrac}

\usepackage{thm-restate} % TS

\newenvironment{proofof}[1]{\begin{proof}[Proof of #1]}{\end{proof}}

 \renewcommand{\epsilon}{\varepsilon}

\newcommand{\poly}{\mathop{\mathrm{poly}}}
\newcommand{\polylog}{\mathop{\mathrm{polylog}}}
\newcommand{\quasipoly}{\mathop{\mathrm{quasipoly}}}
\newcommand{\eps}{\epsilon}

\newcommand{\N}{\mathbb{N}}

\newcommand{\BPP}{\mathsf{BPP}\xspace}

\newcommand{\RE}{\mathsf{RE}\xspace}

\newcommand{\PTIME}{\mathsf{P}\xspace}

\newcommand{\PH}{\mathsf{PH}\xspace}
\newcommand{\PP}{\mathsf{PP}\xspace}
 
\newcommand{\ETIME}{\mathsf{E}\xspace}
\newcommand{\EXP}{\mathsf{EXP}\xspace}

\newcommand{\NP}{\mathsf{NP}\xspace}
\newcommand{\coNP}{\mathsf{coNP}\xspace}
\newcommand{\PSPACE}{\mathsf{PSPACE}\xspace}

\newcommand{\co}{\mathsf{co}}

\newcommand{\NC}{\mathsf{NC}\xspace}
\newcommand{\TIME}{\mathsf{TIME}}

\newcommand{\NTIME}{\mathsf{NTIME}}
\newcommand{\SPACE}{\mathsf{SPACE}}

\newcommand{\RP}{\mathsf{RP}}

\newcommand{\ZPP}{\mathsf{ZPP}}
\newcommand{\ZPSUBEXP}{\mathsf{ZPSUBEXP}}
\newcommand{\Ppoly}{\PTIME_{\sf/poly}}

\newcommand{\ParP}{\oplus\PTIME}

\def \ro {\text{r.o.-}}

\newcommand{\cro}[1]{$c$\text{-}\ro}

\newcommand{\NE}{\mathsf{NE}}

\newcommand{\E}{\mathsf{E}}

\newcommand{\NEXP}{\mathsf{NEXP}\xspace}

\newcommand{\SAT}{\mathsf{SAT}}
\newcommand{\PAL}{\mathsf{PAL}}
\newcommand{\ssat}{\mathsf{\#3SAT}}
\newcommand{\psat}{\mathsf{\oplus 3SAT}}

\newcommand{\DISJ}{\mathsf{DISJ}}

\newcommand{\AC}{\mathsf{AC}}

\newcommand{\SIZE}{\mathsf{SIZE}}

\newcommand{\MCSP}{\mathsf{MCSP}\xspace}

\newcommand{\sMCSP}{\mathsf{search\text{-}MCSP}\xspace}

\newcommand{\GapMAJ}{\mathsf{PromiseMAJORITY}}

\newcommand{\Parity}{\mathsf{Parity}}

\def\tt{\mathsf{tt}}

\def\caC{\mathcal{C}}
\def\caD{\mathcal{D}}
\def\calD{\mathcal{D}}

\newcommand{\RKt}{\mathsf{R_{K^t}}}

\newcommand{\Ksupt}{\mathsf{K^t}}

\title{Constructive Separations and Their Consequences}

\ThCSshortnames{L.~Chen, C.~Jin, R.~Santhanam, and R.~Williams}
 \ThCSshorttitle{Constructive Separations and Their Consequences}
\ThCSauthor[1]{Lijie Chen}{lijieche@mit.edu}[https://orcid.org/0000-0002-6084-4729]

\ThCSauthor[1]{Ce Jin}{cejin@mit.edu}[https://orcid.org/0000-0001-5264-1772]

\ThCSauthor[2]{Rahul Santhanam}{rahul.santhanam@cs.ox.ac.uk}[https://orcid.org/0000-0002-8716-6091]

\ThCSauthor[1]{Ryan Williams}{rrw@mit.edu}[https://orcid.org/0000-0003-2326-2233]
\ThCSthanks{A preliminary version of this article appeared at FOCS 2021 \cite{0001JSW21}. Ryan Williams was supported by NSF CCF-2127597 and CCF-1909429. Part of this work was completed while Ryan Williams was visiting the Simons Institute for the Theory of Computing, participating in the `Theoretical Foundations of Computer Systems' and `Satisfiability: Theory, Practice, and Beyond' programs.}
\ThCSaffil[1]{Massachusetts Institute of Technology, USA}
\ThCSaffil[2]{University of Oxford, UK}

\ThCSkeywords{complexity classes, lower bounds, refuters}

\ThCSyear{2024}
\ThCSarticlenum{3}
\ThCSdoi{10.46298/theoretics.24.3}
\ThCSreceived{Mar 29, 2022}
\ThCSrevised{Jan 18, 2024}
\ThCSaccepted{Jan 24, 2024}
\ThCSpublished{Feb 15, 2024}

\begin{document}
	
	\maketitle 
	
	\begin{abstract}
		For a complexity class ${\cal C}$ and language $L$, a \textbf{constructive separation of} $L \notin \caC$ gives an efficient algorithm (also called a \textbf{refuter}) to find counterexamples (bad inputs) for every~$\caC$ algorithm attempting to decide $L$. We study the questions: Which lower bounds can be made constructive? What are the consequences of constructive separations? We build a case that  ``constructiveness'' serves as a dividing line between many weak lower bounds we know how to prove, and strong lower bounds against $\PTIME$, $\ZPP$, and $\BPP$. Put another way, constructiveness is the opposite of a complexity barrier: it is a property we want lower bounds to have. Our results fall into three broad categories.
		
		\begin{itemize}
			\item \textbf{For many separations, making them constructive would imply breakthrough lower bounds.} Our first set of results shows that, for many well-known lower bounds against streaming algorithms, one-tape Turing machines, and query complexity, as well as lower bounds for the Minimum Circuit Size Problem, making these lower bounds constructive would imply breakthrough separations ranging from $\EXP^\NP \ne \BPP$ to even $\PTIME \ne \NP$. For example, it is well-known that distinguishing binary strings with $(1/2-\eps)n$ ones from strings with $(1/2+\eps)n$ ones requires randomized query complexity $\Theta(1/\eps^2)$. We show that a sufficiently constructive refuter for this query lower bound would imply $\PTIME \neq \NP$.
		
			\item \textbf{Most conjectured uniform separations can be made constructive.} Our second set of results shows that for most major open problems in lower bounds against $\PTIME$, $\ZPP$, and $\BPP$, including $\PTIME \ne \NP$, $\PTIME \ne \PSPACE$, $\PTIME \ne \PP$, $\ZPP \ne \EXP$, and $\BPP \ne \NEXP$, any proof of the separation would further imply a \emph{constructive} separation. Our results generalize earlier results for $\PTIME \ne \NP$ [Gutfreund, Shaltiel, and Ta-Shma, CCC 2005] and $\BPP \ne \NEXP$ [Dolev, Fandina and Gutfreund, CIAC 2013]. Thus any proof of these strong lower bounds must also yield a constructive version, in contrast to many weak lower bounds we currently know.
		
			\item \textbf{Some separations cannot be made constructive.} Our third set of results shows that certain complexity separations cannot be made constructive. We observe that for $t(n) \geq n^{\omega(1)}$ there are no constructive separations for $\RKt$ (which is known to be not in $\PTIME$) from any complexity class, unconditionally. We also show that under plausible conjectures, there are languages in $\NP \setminus \PTIME$ for which there are no constructive separations from any complexity class.
		\end{itemize}
	
	\end{abstract}
 
\section{Introduction}

\newcommand{\theorestated}{} % For theorems that will later be
% restated

A primary goal of complexity theory is to derive strong complexity lower bounds for natural computational problems. When a lower bound holds for a problem $\Pi$ against a model ${\cal M}$ of algorithms, this implies that for each algorithm $A$ from ${\cal M}$, there is an infinite sequence of {\it counterexamples} $\{x_i\}$ for which $A$ fails to solve $\Pi$ correctly.\footnote{If the family of counterexamples was finite, we could hard-code them into the algorithm $A$ to give a new algorithm $A'$ that solves $\Pi$ correctly, for most "reasonable" models ${\cal M}$.} In this paper, we study the question: can such a family of counterexamples be constructed efficiently, for fixed $\Pi$ and a given algorithm $A$ in ${\cal M}$? We call a positive answer to this question a {\it constructive} separation of~$\Pi$ from ${\cal M}$.

There are several motivations for studying this question in a systematic way for natural problems $\Pi$ and models ${\cal M}$. Computer science is inherently a constructive discipline, and it is natural to ask if a given lower bound can be made constructive. Indeed, this can be seen as an ``explicit construction" question of the kind that is studied intensively in the theory of pseudorandomness, where we may have a proof of existence of certain objects with optimal parameters, e.g., extractors, and would like to construct such objects efficiently. %

Our primary motivation is to understand the general lower bound problem better! Constructive lower bounds have led to some recent resolutions of lower bound problems in complexity theory, and we believe they will lead to more. In his Geometric Complexity Theory approach, Mulmuley~\cite{MulmuleyFlip10} suggests that in order to break the ``self referential paradox'' of $\PTIME$ vs $\NP$ and related problems\footnote{Namely, since the $\PTIME$ vs.\ $\NP$ problem is a universal statement about mathematics that says that discovery is hard, why could it not preclude its own proof and hence be independent of the axioms of set theory?}, one has to shoot for \textbf{algorithms} which can efficiently find counterexamples for any algorithms claiming to solve the conjectured hard language. This view has been dominant in the GCT approach towards the $\mathsf{VNP}$ vs.\ $\mathsf{VP}$ problem~\cite{MulmuleyVI,Mulmuley12,IkenmeyerK20}. %

The ability to ``construct bad inputs for a hard function'' has also been critical to some recent developments in (Boolean) complexity theory. Chen, Jin, and Williams~\cite{cjw20} studied a notion of constructive proof they called \emph{explicit obstructions}. %
They show several ``sharp threshold'' results for explicit obstructions, demonstrating (for example) that explicit obstructions unconditionally exist for $n^{2-\eps}$-size DeMorgan formulas, but if they existed for $n^{2+\eps}$-size formulas then one could prove the breakthrough lower bound $\EXP \not\subset \NC^1$. (We discuss the differences between their work and ours in Section~\ref{sec:related-work}, along with other related work.) %

Constructive lower bounds have also been directly useful in proving recent lower bounds. Chen, Lyu, and Williams~\cite{ChenLW20} recently showed how to strengthen several prior lower bounds for $\E^{\NP}$ based on the algorithmic method to hold \emph{almost everywhere}. A key technical ingredient in this work was the development of a \emph{constructive} version of a nondeterministic time hierarchy that was already known to hold almost everywhere~\cite{FortnowS16}. The ``refuter'' in the constructive lower bound (the algorithm producing counterexamples) is used directly in the design of the hard function in $\E^{\NP}$. This gives a further motivation to study when lower bounds can be made constructive.

\paragraph{The Setup.}
Define \emph{typical complexity classes} as the classes $\PTIME$, $\BPP$, $\ZPP$, $\NP$, $\Sigma_k \PTIME$,  $\PP$, $\oplus \PTIME$, $\PSPACE$, $\EXP$, $\NEXP$, $\EXP^\NP$ and their complement classes.

Intuitively, a \emph{refuter} for $f$ against an algorithm $A$ is an algorithm $R$ that finds a counterexample on which $A$ makes a mistake, proving in an algorithmic way that $A$ cannot compute $f$. (This notion seems to have first been introduced by Kabanets~\cite{Kabanets00} in the context of derandomization; see Section~\ref{sec:related-work} for more details.)
\begin{definition}[Refuters and constructive separation]
	For a 
	function $f \colon \{0,1\}^{\star} \rightarrow \{0,1\}$ and an algorithm $A$, a $\PTIME$-\emph{refuter} for $f$ against $A$ is a deterministic polynomial time algorithm $R$ that, given input $1^n$, prints a string $x\in \{0,1\}^n$,  such that for infinitely many $n$,  $A(x)\neq f(x)$.
	
	We extend this definition to randomized refuters as follows:
	\begin{itemize}
		\item A $\BPP$-\emph{refuter} for $f$ against $A$ is a randomized polynomial time algorithm $R$ that, given input $1^n$, prints a string $x\in \{0,1\}^n$,
		such that for infinitely many $n$,  $A(x)\neq f(x)$ with probability at least $2/3$.\footnote{We remark that here it is not necessarily possible to amplify the success probability, so the choice of the constant $2/3$ matters.}
		\item A $\ZPP$-\emph{refuter}
		for $f$ against $A$ is a randomized polynomial time algorithm $R$ that, given input $1^n$, prints $x\in \{0,1\}^n\cup \{\perp\}$, such that for infinitely many $n$, either $x= \perp$ or $A(x)\neq f(x)$ with probability $1$, and $x\neq \perp$ with probability at least $2/3$.
	\end{itemize}

	For ${\cal D}\in \{\PTIME, \BPP, \ZPP\}$ and a typical complexity class ${\cal C}$, we say there is a \emph{${\cal D}$-constructive separation of $f \notin {\cal C}$}, if for every algorithm $A$ computable in ${\cal C}$, there is a refuter for $f$ against $A$ that is computable in ${\cal D}$. 
	\label{defn:refuter}
\end{definition}

Note that we allow the refuter algorithm to depend on the algorithm $A$.  

In \cref{defn:refuter}, we allow the algorithm $A$ being refuted to be randomized, but we only consider randomized algorithms $A$ with bounded probability gap, that is, on every input $x$ there is an answer $b\in \{0,1\}$ such that $A$ outputs $b$ with at least $2/3$ probability, and we denote this answer $b$ by $A(x)$. 

In \cref{defn:refuter} we restrict ${\cal C}$ and ${\cal D}$ to come from a small list of complexity classes in order to be formal and concrete; the definition can of course be generalized naturally to other classes of algorithms.

The length requirement that $|R(1^n)|=n$ in \cref{defn:refuter} is important.\footnote{This requirement seems somewhat strong, but it is easy to show that if there is a refuter which on infinitely many $1^n$ always outputs a string of length in $\{n, n^{c_1},\ldots,n^{c_k}\}$ for some constants $c_1, \ldots,c_k > 0$, then there is another refuter which outputs a string of length $n$ on infinitely many $1^n$. See \cref{sec:uniform-can-constructive}.}
For example, if $x=R(1^n)$ has very short length $|x| = \log(n)$, then the task of refutation would be much easier, as one has exponential $2^{O(|x|)}$ time to produce an input $x$.

Our work is certainly not the first to consider the efficiency of producing ``bad'' inputs for weak algorithms. Gutfreud, Shaltiel, and Ta Shma~\cite{GutfreundST07} showed that if $\PTIME \neq \NP$, then there is a $\PTIME$-constructive separation of $\PTIME \neq \NP$ (in other words, there is a $\PTIME$-constructive separation of SAT $\notin \PTIME$). They also proved analogous results for $\ZPP \neq \NP$ and $\NP \not\subset \BPP$; in these results, they considered the setting where the randomized algorithm being refuted may have unbounded probability gap, which is more general than what we consider in this paper.
Building on the technique of~\cite{GutfreundST07}, Dolev, Fandina and Gutfreund~\cite{DolevFG13} established a $\BPP$-constructive separation of $\BPP \neq \NEXP$ (assuming $\BPP \neq \NEXP$ is true).
They also proved a similar result for $\ZPP \neq \NEXP$.\footnote{We remark that \cite{DolevFG13} only considered refuting algorithms with \emph{one-sided} error.}
Atserias \cite{Atserias06} showed that if $\NP \not \subset \PTIME/\poly$, then there is a $\BPP$-constructive separation of $\NP \not \subset \PTIME/\poly$.

At this point it is natural to ask: 
\begin{quote}
	{\bf Question 1:} Which lower bounds imply a corresponding \emph{constructive} lower bound? 
\end{quote}

Naively, one might expect that the answer to Question 1 is positive when the lower bound is relatively easy to prove. %
We show that this intuition is wildly inaccurate. On the one hand, we show that for many natural examples of problems $\Pi$ and weak models ${\cal M}$, a lower bound is easily provable (and well-known), but {\it constructivizing} the \emph{same} lower bound would imply a breakthrough separation in complexity theory (a much stronger type of lower bound). On the other hand, we show that for many "hard" problems $\Pi$ and strong models ${\cal M}$, a lower bound for $\Pi$ against ${\cal M}$ {\it automatically} constructivizes: the existence of the lower bound alone can be used to derive an algorithm that produces counterexamples. So, in contrast with verbs such as ``relativize''~\cite{BakerGS75}, ``algebrize''~\cite{AaronsonW09}, and ``naturalize''~\cite{RazborovR97}, we \emph{want} to prove lower bounds that constructivize! We are identifying a \emph{desirable} property of lower bounds.

We now proceed to discuss our results in more detail, and then give our interpretation of these results.

\subsection{Most Conjectured Poly-Time Separations Can Be Made Constructive}
Generalizing prior work~\cite{GutfreundST07,DolevFG13}, we show that for most  major open lower bound problems regarding polynomial time, their resolution implies corresponding \emph{constructive} lower bounds for most complete problems.

\newcommand{\footnoterefuteruniformsep}{\footnote{Throughout this paper when we say  a language $L$ is $\caD$-complete, we mean it is $\caD$-complete under polynomial-time many-one reductions. A language $L$ is \emph{paddable} if there is a deterministic polynomial-time algorithm that receives $(x,1^n)$ as input, where string $x$ has length at most $n-1$, and then outputs a string $y\in \{0,1\}^n$ such that $L(x)=L(y)$.}}
\begin{restatable}{theorem}{theorefuteruniformsep}\label{theo:refuter-uniform-sep} \theorestated
	Let $\caC \in \{\PTIME, \ZPP,\BPP \}$ and let $\caD \in \{
        \NP, \Sigma_2 \PTIME, \dotsc ,\Sigma_k \PTIME, \dotsc$,$  \PP,
        \PSPACE$, $\EXP$, $\NEXP$, $\EXP^\NP\}$. Then $\caD \nsubseteq
        \caC$ implies that for every paddable $\caD$-complete language
        $L$, there is a $\caC$-constructive separation of $L \notin
        \caC$.\footnoterefuteruniformsep\ 
	Furthermore, $\ParP \nsubseteq \caC$ implies that for every paddable $\ParP$-complete language $L$, there is a $\BPP$-constructive separation of $L \notin \caC$.
\end{restatable}

In other words, for many major separation problems such as $\PP \neq \BPP$, $\EXP \neq \ZPP$, and $\PSPACE \neq \PTIME$, proving the separation automatically implies constructive algorithms that can produce counterexamples to any given weak algorithm. We find Theorem~\ref{theo:refuter-uniform-sep} to be mildly surprising: intuitively it seems that proving a constructive lower bound should be strictly stronger than simply proving a lower bound. (Indeed, we will later see other situations where making \emph{known} lower bounds constructive would have major consequences!) Moreover, for separations beyond $\PTIME \neq \NP$, the polynomial-time refuters guaranteed by Theorem~\ref{theo:refuter-uniform-sep} are producing hard instances for  problems that presumably do not have short certificates. For example, we do not believe that $\PSPACE = \NP$ (we do not believe $\PSPACE$ has short certificates), yet one can refute polynomial-time algorithms attempting to solve QBF with other polynomial-time algorithms, under the assumption that $\PSPACE \neq \PTIME$. The point is that such polynomial-time refuters intuitively cannot check their own outputs for correctness. We find this very counterintuitive.

\subsection{Unexpected Consequences of Making Some Separations Constructive}

Given Theorem~\ref{theo:refuter-uniform-sep}, we see that most of the major open problems surrounding polynomial-time lower bounds would yield constructive separations. Can \emph{all} complexity separations be made constructive? It turns out that for several ``weak'' lower bounds proved by well-known methods, making them constructive requires proving \emph{other} breakthrough lower bounds! %

Thus, there seems to be an algorithmic ``dividing line'' between many lower bounds we are able to prove, and many of the longstanding lower bounds that seem perpetually out of reach. The longstanding separation questions (as seen in Theorem~\ref{theo:refuter-uniform-sep}) \emph{require} a constructive proof: an efficient algorithm that can print counterexamples. Here we show that many lower bounds we are able to prove do not require constructivity, but if they could be made constructive then we would prove a longstanding separation! In our minds, these results confirm the intuition of Mulmuley that we should ``go for explicit proofs'' in order to make serious progress on lower bounds, especially uniform ones.

\paragraph*{Constructive Separations for (Any) Streaming Lower Bounds Imply Breakthroughs.} It is well-known that various problems are \emph{unconditionally} hard for low-space randomized streaming algorithms. For example, from the randomized communication lower bound for the Set-Disjointness ($\DISJ$) problem~\cite{KalyanasundaramS92,Razborov92,Bar-YossefJKS04}, it follows that no $n^{1-\eps}$-space randomized streaming algorithm can solve $\DISJ$ on $2n$ input bits.\footnote{Recall in the $\DISJ$ problem, Alice is given an $n$-bit string $x$, Bob is given an $n$-bit string $y$, and the goal is to determine whether their inner product $\sum_{i=1}^{n} x_i y_i$ is nonzero.}

Clearly, every $n^{o(1)}$-space streaming algorithm for $\DISJ$ \emph{must} fail to compute $\DISJ$ on some input (indeed, it must fail on many inputs). We show that efficient refuters against streaming algorithms attempting to solve \emph{any} $\NP$ problem would imply a breakthrough lower bound on \emph{general} randomized algorithms, not just streaming algorithms.

Similarly to \cref{defn:refuter}, we can consider
$\PTIME^{\NP}$-refuters against streaming algorithms, which are
deterministic polynomial time algorithms given a SAT oracle that
output counterexamples for streaming algorithms on infinitely many
input lengths. We can also similarly define
$\PTIME^{\NP}$-constructive separations.
\newcommand{\footnoteuniformPNP}{\footnote{That is, for every such randomized streaming algorithm $A$, there is a $\PTIME^{\NP}$ refuter $B$ such that $B(1^n)$ prints an input $x$ of length $n$ such that $A$ decides whether $x \in L$ incorrectly, for infinitely many $n$.}}
\begin{restatable}{theorem}{theoPNP}\label{theo:PNP-constructive-to-EXPNP-lowb} \theorestated
	Let $f(n) \ge \omega(1)$. For every language $L \in \NP$, a $\PTIME^\NP$-constructive separation of~$L$ from uniform randomized streaming algorithms with $O(n \cdot (\log n)^{f(n)})$ time and $O(\log n)^{f(n)}$ space\footnoteuniformPNP\ implies $\EXP^{\NP} \neq \BPP$.
\end{restatable}
\renewcommand{\footnoteuniformPNP}{}
  
Essentially every lower bound proved against streaming algorithms in the literature holds for a problem whose decision version is  in $\NP$. 
Theorem~\ref{theo:PNP-constructive-to-EXPNP-lowb} effectively shows that if \emph{any} of these lower bounds can be made constructive, even in a $\PTIME^{\NP}$ sense, then we would separate randomized polynomial time from $\EXP^{\NP}$, a longstanding open problem in complexity theory. 
We are effectively showing that the counterexamples printed by such a refuter algorithm must encode a function that is hard for \emph{general} randomized streaming algorithms.

Stronger lower bounds follow from more constructive refuters (with an algorithm in a lower complexity class than $\PTIME^{\NP}$) against randomized streaming algorithms. At the extreme end, we find that uniform circuits refuting $\DISJ$ against randomized streaming algorithms would even imply $\PTIME \ne \NP$.
Similarly to \cref{defn:refuter}, we can consider polylogtime-uniform-$\AC^0$-refuters against streaming algorithms, which are polylogtime-uniform-$\AC^0$ circuits that output counterexamples for streaming algorithms on infinitely many input lengths. 

\newpage % TS Forced new page to avoid footnote text to appear on wrong page 
\newcommand{\footnoteuniformAC}{\footnote{That is, for every such randomized streaming algorithm $A$, there is a polylogtime-uniform $\AC^0$ circuit family $\{C_n\}$ such that $A$ fails to solve $\DISJ$ on $2n$-bit inputs correctly on the output $C_n(1^n)$ for infinitely many $n$.}}
\begin{restatable}{theorem}{theouniformAC}\label{theo:uniformAC0-constructive-to-PNP-lowb} \theorestated
	Let $f(n) \ge \omega(1)$. A polylogtime-uniform-$\AC^0$-constructive separation of $\DISJ$ from randomized streaming algorithms with $O(n \cdot (\log n)^{f(n)})$ time and $O(\log n)^{f(n)}$ space\footnoteuniformAC\ implies $\PTIME \neq \NP$.
      \end{restatable}
\renewcommand{\footnoteuniformAC}{}

To recap, it is well-known that $\DISJ$ does not have randomized streaming algorithms with $O(n \cdot (\log n)^{f(n)})$ time and $O(\log n)^{f(n)}$ space, even for $f(n) \leq o(\log n/ \log \log n)$, by communication complexity arguments. We are saying that, if (given the code of such an algorithm) we can efficiently construct hard instances of $\DISJ$ for that algorithm, then strong lower bounds follow. \emph{That is, making communication complexity arguments constructive would imply strong unconditional lower bounds.} 
\paragraph*{Constructive Separations for One-Tape Turing Machines Imply Breakthroughs.} Next, we show how making some rather old lower bounds constructive would imply a circuit complexity breakthrough. It has been known at least since Maass~\cite{pal} that nondeterministic one-tape Turing machines require $\Omega(n^2)$ time to simulate nondeterministic multitape Turing machines. However, those lower bounds are proved by non-constructive counting arguments. We show that if there is a $\PTIME^{\NP}$ algorithm that can produce bad inputs for a given one-tape Turing machine, then $\E^{\NP}$ requires exponential-size circuits. This in turn would imply $\BPP \subseteq \PTIME^{\NP}$, a breakthrough simulation of randomized polynomial time.

\begin{restatable}{theorem}{theogeneral} \label{thm:general} \theorestated
	For every language $L$ computable by a nondeterministic $n^{1+o(1)}$-time RAM, a $\PTIME^{\NP}$-constructive separation of $L$ from nondeterministic $O(n^{1.1})$-time one-tape Turing machines implies $\E^{\NP} \not \subset \SIZE[2^{\delta n}]$ for some constant $\delta >0$.  
\end{restatable}

\paragraph*{Constructive Separations for Query Lower Bounds Imply Breakthroughs.} Now we turn to query complexity. Consider the following basic problem $\GapMAJ_{n,\eps}$ for a parameter $\eps < 1/2$. 
\begin{quote}
	$\GapMAJ_{n,\eps}$: \emph{Given an input $x\in \{0,1\}^n$, letting $p=\frac{1}{n}  \sum_{i=1}^n x_i$, distinguish between the cases $p<1/2-\eps$ or $p>1/2+\eps$.}
\end{quote}
This is essentially the ``coin problem''~\cite{BrodyV10}. It is well-known that every randomized query algorithm needs $\Theta(1/\eps^2)$ queries to solve $\GapMAJ_{n,\eps}$ with constant success probability (uniform random sampling is the best one can do). That is, any randomized query algorithm making $o(1/\eps^2)$ queries must make mistakes on some inputs, with high probability. We show that constructing efficient refuters for this simple sampling lower bound would imply $\PTIME \ne \NP$!

\begin{restatable}{theorem}{theorefuterqc}\label{theo:refuter-qc-to-PNP} \theorestated
	Let $\eps$ be a function of $n$ satisfying $\eps \leq 1/(\log n)^{\omega(1)}$, and $1/\eps$ is a positive integer computable in $\poly(1/\eps)$ time given $n$ in binary.
	\begin{itemize}
		\item If there is a polylogtime-uniform-$\AC^0$-constructive separation of $\GapMAJ_{n,\eps}$ from randomized query algorithms $A$ using $o(1/\eps^2)$ queries and $\poly(1/\eps)$ time, then $\NP \ne \PTIME$.
		\item If there is a polylogtime-uniform-$\NC^1$-constructive separation of $\GapMAJ_{n,\eps}$ from randomized query algorithms $A$ using $o(1/\eps^2)$ queries and $\poly(1/\eps)$ time, then $\PSPACE \ne \PTIME$.
	\end{itemize}
\end{restatable}
Note that $\GapMAJ_{n,\eps}$ can be easily computed in $\NC^1$. If for every randomized query algorithm $A$ running in $n^{\alpha}$ time and making $n^{\alpha}$ queries for some $\alpha > 0$, we can always find inputs in $\NC^1$ on which $A$ makes mistakes, then we would separate $\PTIME$ from $\PSPACE$.

\paragraph*{Constructive Separations for $\MCSP$ Against $\AC^0$ Imply Breakthroughs.} Informally, the Minimum Circuit Size Problem ($\MCSP$) is the problem of determining the circuit complexity of a given $2^n$-bit truth table. Recent results on the phenomenon of hardness magnification~\cite{OliveiraS18,McKayMW19,cjw19,ChenHOPRS20,cjw20} show that, for various restricted complexity classes $\caC$:
\begin{itemize}
	\item Strong lower bounds against $\caC$ are known for explicit languages.
	\item Standard complexity-theoretic hypotheses imply that such lower bounds should hold also for MCSP (and its variants). 
	\item However, actually proving that $\MCSP \notin \caC$ would imply a breakthrough complexity separation. 
	\item There is also often a slightly weaker lower bound against $\caC$ that can be shown for $\MCSP$, suggesting that we are quantitatively ``close'' to a breakthrough separation in some sense.
\end{itemize}

The scenario where all four conditions above hold is called a ``hardness magnification frontier'' in \cite{ChenHOPRS20}.
We show that a similar phenomenon holds for constructive separations. %
It is well known that versions of $\MCSP$ are not in $\AC^0$~\cite{allenderbkmr06}, but strongly constructive separations are not known. We show that strongly constructive separations would separate $\PTIME$ from $\NP$, and that they exist under a standard complexity hypothesis. Moreover, we show that slightly weaker constructive separations {\it do} exist, and the strong constructive separations we seek do hold for other hard problems such as $\Parity$.

In the following, $\MCSP[s(n)]$ is the computational problem that asks whether a Boolean function on $n$ bits, represented by its truth table, has circuits of size at most $s(n)$.
Similarly to \cref{defn:refuter}, a \emph{polylogtime-uniform-$\AC^0[f(n)]$-refuter for $\MCSP[s(n)]$ against an algorithm $A$} is defined as a polylogtime-uniform-$\AC^0$ circuit of size $f(n)$ that outputs a string $x\in \{0,1\}^N$ given input $1^N$ (where $N=2^n$), such that for infinitely many $N=2^n$, $A(x) \neq \MCSP[s(n)](x)$.\footnote{Note that here we restrict the input lengths $N$ to be powers of two, since otherwise the $\MCSP$ problem is ill-defined.}  
We also consider a natural relaxation of refuter (\cref{defn:refuter}), called \emph{list-refuter}, which outputs a list of $n$-bit strings $x_i$ (instead of a single $n$-bit string $x$) given input $1^n$, and we only require that at least one of the strings $x_i$ is a counterexample.

\begin{restatable}{theorem}{theomag} \label{thm:mag_refuter} \theorestated
	Let $s(n) \ge n^{\log(n)^{\omega(1)}}$ be any time-constructive super-quasipolynomial function.
	In the following, we consider $\MCSP[s(n)]$ and Parity problems of input length $N=2^n$.
	The following hold:
	\begin{enumerate}
		\item (Major Separation from Constructive Lower Bound) If there exists a polylogtime-uniform $\AC^0[\quasipoly]$ refuter for $\MCSP[s(n)]$ against every polylogtime-uniform $\AC^0$ algorithm, then $\PTIME \ne \NP$.
		\item (Constructive Lower Bound Should Exist) If $\PH \not \subseteq \SIZE(s(n)^2)$, then there is a polylogtime-uniform-$\AC^0[\quasipoly]$ refuter for $\MCSP[s(n)]$ against every polylogtime-uniform $\AC^0$ algorithm.
		\item (Somewhat Constructive Lower Bound) For $s(n) \le o(2^n/n)$, there is a polylogtime-uniform-$\AC^0[2^{\poly(s(n))}]$ refuter for $\MCSP[s(n)]$ against every polylogtime-uniform $\AC^0$ algorithm.
		\item (Constructive Lower Bound for a Different Hard Language) There is a  polylogtime-uniform-$\AC^0[\quasipoly]$-list-refuter for $\Parity$ against every polylogtime-uniform $\AC^0$ algorithm.
	\end{enumerate}
\end{restatable}

Note that in item 3, the input size $N$ to the problem is $N=2^n$, hence $2^{\poly(s(n))}$ is only slightly super-quasipolynomial in $N$.

\paragraph{Comparison with Theorem~\ref{theo:refuter-uniform-sep}.} It is very interesting to contrast Theorem~\ref{theo:refuter-uniform-sep} with the various theorems of this subsection.  Theorem~\ref{theo:refuter-uniform-sep} tells us that many longstanding open problems in lower bounds would automatically imply constructive separations, when resolved. 
In contrast, theorems from this subsection say that
extending simple and well-known lower bounds to become constructive would resolve other major lower bounds! Taken together, we view the problem of understanding which lower bounds can be made constructive as a significant key to understanding the future landscape of complexity lower bounds.

\subsection{Certain Lower Bounds Cannot Be Made Constructive}

Finally, we can give some negative answers to our Question 1. %
We show that for some hard functions, there are \emph{no} constructive separations from \emph{any} complexity classes. Specifically, we show (unconditionally or under plausible complexity conjectures) that there are no refuters for these problems against a trivial decision algorithm that \emph{always returns the same answer} (zero, or one). Hence, there can be no constructive separations of these hard languages from any complexity class containing the constant zero or constant one function. (All complexity classes that we know of contain both the constant zero and one function.)

For a string $x \in \{0,1\}^*$, the $t$-time-bounded Kolmogorov complexity of $x$, denote by $\Ksupt(x)$, is defined as the length of the shortest program prints $x$ in time $t(|x|)$. We use $\RKt$ to denote the set of strings $x$ such that $\Ksupt(x) \ge |x| - 1$. Hirahara~\cite{Hirahara20} recently proved that for any super-polynomial $t(n) \geq n^{\omega(1)}$, $\RKt \notin \PTIME$. We observe that this separation cannot be made $\PTIME$-constructive. 

\begin{restatable}{proposition}{propnorefuter}\label{prop:norefuter-RKt} \theorestated
	For any $t(n) \geq n^{\omega(1)}$, there is no $\PTIME$-refuter for $\RKt$ against the constant zero function.
\end{restatable}

Since $\RKt$ is a function in $\EXP$, it would be interesting to find functions in $\NP$ with no constructive separations.\footnote{Note that $\RKt$ is in $\co\NTIME[t(n)]$, but it is likely not in $\coNP$.} We show that under plausible conjectures, such languages in $\NP$ exist.

\newcommand{\footnotenonpa}{\footnote{Here, $\E = \TIME[2^{O(n)}]$, the class of languages decidable in (deterministic) $2^{O(n)}$ time, and $\NE$ is the corresponding nondeterministic class.}}
\newcommand{\footnotenonpb}{\footnote{Here, $\RE = {\sf RTIME}[2^{O(n)}]$, the class of languages decidable in randomized $2^{O(n)}$ time with one-sided error.}}
\begin{restatable}{theorem}{theononp}
	\label{thm:no-np-refuter} \theorestated
	The following hold:
	\begin{itemize}
		\item 
		If $\NE \neq \E$, then there is a language in $\NP \setminus \PTIME$ that does not have $\PTIME$ refuters against the constant one function.\footnotenonpa
		\item 
		If $\NE \neq \RE$, then there is a language in $\NP \setminus \PTIME$ that does not have $\BPP$ refuters against the constant one function.\footnotenonpb
	\end{itemize}
\end{restatable}
\renewcommand{\footnotenonpa}{}

Thus, under natural conjectures about exponential-time classes, there are some problems in $\NP$ with no constructive separations at all, not even against the trivial algorithm that always accepts.

\subsection{Intuition}

Let us briefly discuss the intuition behind some of our results. We will first focus on the results showing that constructive separations of known lower bounds would imply complexity breakthroughs, as we believe these are the most interesting of our paper. 

\paragraph{Constructive Separations of Known Lower Bounds Imply Breakthroughs.} Suppose for example we want to show that a constructive separation of $\SAT$ from quick low-space streaming algorithms implies $\EXP^{\NP} \neq \BPP$. The proof is by contradiction: assuming $\EXP^\NP = \BPP$, we aim to construct a streaming algorithm running in $n (\log n)^{\omega(1)}$ time and $(\log n)^{\omega(1)}$ space which solves $3\SAT$ correctly on all instances produced by $\PTIME^\NP$ algorithms. Given a $\PTIME^\NP$ algorithm $R$, $\EXP^\NP = \BPP$ implies $\EXP^\NP \subset \Ppoly$, which further implies that the output of $R(1^n)$ must have circuit complexity at most $\polylog(n)$ (construed as a truth table).

Extending work of McKay, Murray, and Williams~\cite{McKayMW19}, we show that $\NP \subset \BPP$ (implied by $\EXP^\NP = \BPP$) implies there is an $n (\log n)^{\omega(1)}$ time and $(\log n)^{\omega(1)}$ space randomized algorithm with one-sided error for finding a $\polylog(n)$-size circuit encoding the given length-$n$ input, if such a circuit exists. So given any input $R(1^n)$ from a potential refuter $R$, our streaming algorithm can first compute a $\polylog(n)$-size circuit $C$ encoding $R(1^n)$, and it construes this circuit $C$ as an instance of the $\mathsf{Succinct\text{-}3SAT}$ problem. Since $\mathsf{Succinct\text{-}3SAT} \in \NEXP = \BPP$, our streaming algorithm can solve $\mathsf{Succinct\text{-}3SAT}(C)$ in $\polylog(n)$ randomized time, which completes the proof.

For our results on constructive query lower bounds, we use ideas from learning theory. Set $\eps \ll 1/\poly(\log n)$. Assuming $\PSPACE = \PTIME$, we want to show that for every $n$-bit string printed by an uniform $\NC^1$ circuit $C$ on the input $1^n$, we can decide the $\GapMAJ_{n,\eps}$ problem with $o(1/\eps^2)$ randomized queries in $\poly(1/\eps)$ time. (Then, any sufficiently constructive lower bound that $\GapMAJ_{n,\eps}$ requires $\Omega(1/\eps^2)$ queries would imply $\PTIME \neq \PSPACE$.) $\PSPACE = \PTIME$ implies that for every uniform $\NC^1$ circuit $C$, its output can be encoded by some $\polylog(n)$-size circuit $D$. Now, also assuming $\PSPACE = \PTIME$, this circuit $D$ can be PAC-learned with error $\eps/2$ and failure probability $1/10$ using only $\poly\log(n)/\eps$ queries (and randomness). Let $D'$ be the circuit we learnt through this process; it approximates $D$ well enough that we can make $O(1/\eps^2)$ random queries \emph{to the circuit $D'$, without querying $D$} in $\poly(1/\eps,\log n)$ time, and return the majority answer as a good answer for the original $n$-bit answer. Such an algorithm only makes $\polylog(n)/\eps \ll o(1/\eps^2)$ queries to the original input and runs in $\poly(1/\eps)$ time.

\paragraph*{Constructive Separations for Uniform Complexity Separations.} Next, we highlight some insights behind the proof of~\Cref{theo:refuter-uniform-sep}. The proof is divided into several different cases (Theorem~\ref{thm:list-nexp},~Theorem~\ref{thm3.5}, and~Theorem~\ref{thm:pp}), and we will focus on the intuition behind Theorem~\ref{thm3.5}, which applies to all complexity classes with a downward self-reducible complete language (such as $\PSPACE$ or $\Sigma_k \PTIME$).

We take the $\PSPACE$ vs.\ $\PTIME$ problem as an example. Gutfreund, Shaltiel, and Ta-Shma~\cite{GutfreundST07} showed how to construct refuters for $\PTIME \neq \NP$, but their proof utilizes the search-to-decision reduction for $\NP$-complete problems, and no such reduction exists for $\PSPACE$. We show how a downward self-reduction can be used to engineer a situation similar to that of~\cite{GutfreundST07}.

Let $M$ be a downward self-reducible $\PSPACE$-complete language and let $A$ be a $\PTIME$ algorithm. We also let $D$ be a polynomial-time algorithm defining a downward-self reduction for $M$, so that for all but finitely many $n \in \N$ and $x \in \{0,1\}^n$,
\begin{equation}
	D(x)^{M_{\le n-1}} = M(x).\label{eq:define-M}
\end{equation} That is, $D$ can compute $M(x)$ given access to an $M$-oracle for all strings of length less than $|x|$.
Our key idea is that~\eqref{eq:define-M} \emph{also} defines $M$. Assuming the polynomial-time algorithm $A$ cannot compute $M$, it follows that~\eqref{eq:define-M} does not always hold if $M$ is replaced by $A$. In particular, the following $\NP$ statement is true for infinitely many $n$: %
\begin{equation}
	\text{$\exists x \in \{0,1\}^n$ such that }
	D(x)^{A_{\le n-1}} \ne A(x). %
	\label{eq:NP-statement}
\end{equation}

Now we use a similar approach as in~\cite{GutfreundST07}: we use $A$ and a standard search-to-decision reduction to find the shortest string $x^*$ so that~\eqref{eq:NP-statement} holds. If $A$ fails to do so, we can construct a counterexample to the claim that $A$ solves the $\PSPACE$-complete language $M$ similarly to~\cite{GutfreundST07}. If~$A$ finds such an $x^*$, then by definition $A(y) = M(y)$ for all $y$ with $|y| \le |x^*|-1$ and we have $A(x^*) \ne M(x^*)$ from~\eqref{eq:NP-statement}, also a counterexample.\footnote{Note the argument above only finds a single counterexample; using a paddable $\PSPACE$-complete language, one can adapt the above argument to find infinitely many counter examples, see the proof of~\Cref{thm3.5} for details.}

\subsection{Organization}

In Section~\ref{sec:prelim} we introduce the necessary definitions and technical tools for this paper, as well as review other related work. 
In Section~\ref{sec:refuter-for-streaming-algorithms} we show that making known streaming and query lower bounds constructive implies major complexity separations, and prove Theorem~\ref{theo:PNP-constructive-to-EXPNP-lowb} and Theorem~\ref{theo:uniformAC0-constructive-to-PNP-lowb}.
In Section~\ref{sec:constructive-sep-MCSP} we show that certain constructive separations for $\MCSP$ imply breakthrough lower bounds such as $\PTIME \ne \NP$, and prove Theorem~\ref{thm:mag_refuter}.
In Section~\ref{sec:uniform-can-constructive} we study constructive separations for uniform classes and prove Theorem~\ref{theo:refuter-uniform-sep}.
In Section~\ref{sec:hard-language-no-constructive-sep} we show that several hard languages do not have constructive separations from any complexity class, and prove Proposition~\ref{prop:norefuter-RKt} and Theorem~\ref{thm:no-np-refuter}. 
Finally, in Section~\ref{sec:future-work} we conclude with some potential future work.

\renewcommand{\theorestated}{{\normalfont\bfseries (Restated)\ }} % Next
                                % time "restatable" theorems will  be
% qualified as restated

\section{Preliminaries}\label{sec:prelim}
\subsection{Notation}
We use $\widetilde O(f)$ as shorthand for $O(f\cdot \polylog (f))$ throughout the paper. All logarithms are base-2. We use $n$ to denote the number of input bits.
We say a language $L \subseteq \{0,1\}^{\star}$ is $f(n)$-sparse if $|L_n| \le f(n)$, where $L_n = L \cap \{0,1\}^n$.
We assume knowledge of basic complexity theory (see~\cite{AB09-book,GoldreichBook08}).

	\subsection{Definitions of MCSP and time-bounded Kolmogorov complexity}\label{sec:prelim-MCSP-MKtP}

	The Minimum Circuit Size Problem ($\MCSP$)~\cite{kabanets-cai00} and  $t$-time-bounded Kolmogorov complexity ($\Ksupt$) are studied in this paper. We recall their definitions. 

	\begin{definition}[$\MCSP$] Let $s\colon \N \to \N$ satisfy $s(n)\ge n-1$ for all $n$.  
		
		Problem: $\MCSP[s(n)]$.
		
		Input: A function $f\colon\{0,1\}^n\to \{0,1\}$, presented as a truth table of $N= 2^n$ bits.
		
		Decide: Does $f$ have a (fan-in two) Boolean circuit  $C$ of size at most $s(n)$?
	\end{definition}

    We will also consider $\sMCSP$, the search version of $\MCSP$, in which the small circuit~$C$ must be output when it exists.

	For a time bound $t\colon \mathbb{N} \to \mathbb{N}$, recall that the $\mathsf{K^t}$ complexity ($t$-time-bounded Kolmogorov complexity) of string $x$ is the length of the shortest program which outputs $x$ in at most $t(|x|)$ time.
	
	\begin{definition}[$\RKt$] Let $t\colon \N \to \N$.
		
		Problem: $\RKt$.
		
		Input: A string $x \in \{0,1\}^n$.
		
		Decide:  Does $x$ have $\Ksupt(x)$ complexity at least $n - 1$?
		\label{defi:RKt}
	\end{definition}

\subsection{Implications of Circuit Complexity Assumptions on Refuters}

The following technical lemma shows that, assuming uniform classes have non-trivially smaller circuits, the output of a refuter may be assumed to have low circuit complexity. This basic fact will be useful for several proofs in the paper.

\begin{lemma}\label{lemma:P-NP-small-circuit}
Let $s\colon \N \to \N$ be an increasing function. The following hold:

\begin{enumerate}
    \item Assuming $\E^\NP \subset \SIZE[s(n)]$, then for every $\PTIME^\NP$ algorithm $R$ such that $R(1^n)$ outputs $n$ bits, it holds that $R(1^n)$ has circuit complexity at most $s(O(\log n))$.
    \item Assuming $\E \subset \SIZE[s(n)]$, then for every $\PTIME$ algorithm $R$ such that $R(1^n)$ outputs $n$ bits, it holds that $R(1^n)$ has circuit complexity at most $s(O(\log n))$.
    \item Assuming $\SPACE[O(n)] \subset \SIZE[s(n)]$, then for every LOGSPACE algorithm $R$ such that $R(1^n)$ outputs $n$ bits, it holds that $R(1^n)$ has circuit complexity at most $s(O(\log n))$.
\end{enumerate}
\end{lemma}
\begin{proof}
In the following we only prove the first item, the generalization to the other two items are straightforward.

Consider the following function $f_R(n,i)$, which takes two binary integers $n$ and $i \in [n]$ as inputs, and output the $i$-th bit of the output of $R(1^n)$. The inputs to $f_R$ can be encoded in $O(\log n)$ bits in a way that all inputs $(n,i)$ with the same $n$ has the same length. 

Since $R$ is in $\PTIME^\NP$, we have $f_R \in \E^\NP$. By our assumption and fix the first part of the input to~$f_R$ as $n$, it follows that $R(1^n)$ has circuit complexity at most $s(O(\log n))$.
\end{proof}

The following simple corollary of Lemma~\ref{lemma:P-NP-small-circuit} will also be useful.

\begin{corollary}\label{cor:small-circuit-special-case}
    If $\ETIME^\NP \subset \Ppoly$ ($\ETIME \subset \Ppoly$ or $\SPACE[O(n)] \subseteq \Ppoly$), then for every $\PTIME^\NP$ ($\PTIME$ or LOGSPACE) algorithm $R$ such that $R(1^n)$ outputs $n$ bits, it holds that $R(1^n)$ has circuit complexity at most $\polylog(n)$.
\end{corollary}

We also observe that $\PTIME = \NP$ has strong consequences for polylogtime-uniform $\AC^0$ circuits.
\begin{lemma} \label{lem:pnpac0} 
	The following hold:
   \begin{enumerate}
       \item 
   Assuming $\PTIME = \NP$, then for every polylogtime-uniform $\AC^0$ algorithm $R$ such that $R(1^n)$ outputs $n$ bits, it holds that $R(1^n)$ has circuit size complexity at most $\polylog(n)$.
   \item 
   Assuming $\PTIME = \PSPACE$, then for every polylogtime-uniform $\NC^1$ algorithm $R$ such that $R(1^n)$ outputs $n$ bits, it holds that $R(1^n)$ has circuit size complexity at most $\polylog(n)$.
   \end{enumerate}
\end{lemma}
\begin{proof}
Let $B$ be a polylogtime-uniform algorithm that, on the integer $n$ (in binary) and $O(\log n)$-bit additional input, reports gate and wire information for an $\AC^0$ circuit $R_n$. Consider the function $f(n,i)$ which determines the $i$-th output bit of the circuit $R_n$ on the input $1^n$, given~$n$ and~$i$ in binary. The function $f$ is a problem in $\PH$: given input of length $m=O(\log n)$, by existentially and universally guessing and checking gate/wire information (and using the $\polylog (n)$-time algorithm $B$ to verify the information), the $R_n$ of $n^{O(1)}$ size can be evaluated in $\Sigma_d \TIME[m^k]$ for a constant $d$ depending on the depth of $R_n$, and a constant $k$ depending on the algorithm $B$.
Since $\PTIME = \NP$, $f$ is computable in $\PTIME$, i.e.,  $f$ is in time at most $\alpha m^\alpha$ for some constant~$\alpha$ depending on $k$, $d$, and the polynomial-time $\SAT$ algorithm. Therefore $f$ has a circuit family of size at most $m^c$ for some fixed $c$, where $m = O(\log n)$.
Thus the output of such a family always has small circuits. 

The same argument applies if we replace $\AC^0$ by $\NC^1$ and replace $\PH$ by $\PSPACE$.
\end{proof}

\subsection{Other Related Work}
\label{sec:related-work}

Beyond the prior work on efficient refuters stated in the introduction (such as~\cite{GutfreundST07,Atserias06,DolevFG13}), other work on efficient methods for producing hard inputs includes~\cite{LiptonY94,Gutfreund06, BogdanovTW10,Vereshchagin13,OliveiraS18}).

As mentioned in the introduction, Kabanets~\cite{Kabanets00} defined and studied refuters in the context of derandomization. A primary result from that paper is that it is possible to simulate one-sided error polynomial time ($\RP$) in zero-error subexponential time ($\mathsf{ZPSUBEXP}$) on all inputs produced by refuters (efficient time algorithms that take $1^n$ and output strings of length $n$).\footnote{The exact statement involves an ``infinitely-often'' qualifier, which we omit here for simplicity. A version of the simulation that removes the restriction to refuters, with the addition of a small amount of advice, was given in~\cite{Williams16Derand}.} In other words, nontrivial derandomization is indeed possible when we only consider the outputs of refuters: there is {\bf no} constructive separation of $\RP \not\subset \ZPSUBEXP$. This result contrasts nicely with some of our own, which show that if we could prove (for example) $\EXP = \ZPP$ holds with respect to refuters, then $\EXP = \ZPP$ holds unconditionally. (Of course this is a contrapositive way of stating our results; we don't believe that $\EXP = \ZPP$ holds!) Kabanets' work effectively shows that if $\RP \not\subseteq \ZPSUBEXP$ implied a \emph{constructive separation} of $\RP \not\subseteq \ZPSUBEXP$, then $\RP \subseteq \ZPSUBEXP$ holds unconditionally (because there is no constructive separation of $\RP$ from $\ZPSUBEXP$). Other works in this direction include~\cite{ImpagliazzoW01,TrevisanV07,Lu01,GutfreundST03,ShaltielU09,ChenT21b,ChenT23}.

Chen, Jin, and Williams~\cite{cjw20} studied a notion of constructive proof they called \emph{explicit obstructions}. Roughly speaking, an explicit obstruction against a circuit class $\mathcal{C}$ is a (deterministic) polynomial-time algorithm $A$ outputting a list $L_n$ of input/output pairs $\{(x_i,y_i)\}$ with distinct~$x_i$, such that all circuits in $\caC$ fail to be consistent on at least one input/output pair. Chen, Jin, and Williams show several ``sharp threshold'' results for explicit obstructions, demonstrating (for example) that explicit obstructions unconditionally exist for $n^{2-\eps}$-size DeMorgan formulas, but if they existed for $n^{2+\eps}$-size formulas then one could prove the breakthrough lower bound $\EXP \not\subset \NC^1$. In this work, we are considering a ``uniform'' version of this concept: instead of outputting a list of bad input/output pairs (that do not depend on the algorithm), here we only have to output one bad instance that depends on the algorithm given.

An additional motivation for studying constructive proofs comes from proof complexity and bounded arithmetic. A circuit lower bound for a language $L \in \PTIME$ can naturally be expressed by a $\Pi_2$ statement $S_n$ that says: "For all circuits $C$ of a certain type, there exists $x$ of length~$n$ such that $C(x) \neq L(x)$". In systems of bounded arithmetic such as Cook's theory $PV_1$~\cite{Cook75} (formalizing poly-time reasoning) or Je\v r\'abek's theory $APC_1$~\cite{Jerabek07} (formalizing probabilistic poly-time reasoning), a proof of $S_n$ for infinitely many $n$ immediately implies a constructive separation. The reason is that these theories have efficient witnessing: informally, any proof of a ``$\forall\exists$-statement'' $\forall x \exists y R(x,y)$ (for polynomial-time computable $R$) in these theories constructs an efficiently computable function $f$ such that $R(x, f(x))$ holds. Here the function $f$ plays the role of the refuter in a constructive separation. Therefore, situations in which constructive separations are unlikely to exist may provide clues about whether complexity lower bounds could be independent of feasible theories. Conversely, the constructiveness of a separation is a precondition for the provability of that separation in these feasible theories.\footnote{We note, however, that these connections depend on the complexity classes being separated. A circuit lower bound for an $\NP$ problem does not have an obvious $\Pi_2$ formulation, so the efficient witnessing results mentioned above do not directly apply. More complicated witnessing theorems might still be relevant; we refer to \cite{Pich15}, \cite{MullerPich20}, and the recent book on Proof Complexity by Kraj\'i\v cek~\cite{krajivcek2019proof} for a more detailed discussion of these matters.}

\paragraph{Hardness Magnification.} Another related line of work is hardness magnification~\cite{OliveiraS18,McKayMW19,OliveiraPS18,ChenHOPRS20}. This line of work shows how very minor-looking lower bounds actually hide the whole difficulty of $\PTIME$ vs $\NP$ and related problems. However, one might say that those results simply illuminate large holes in our intuition: those minor-looking lower bounds are far more difficult to prove than previously believed. One has to be skeptical in considering hardness magnification as a viable lower bounds approach, because we really don't understand how difficult the ``minor-looking'' lower bounds actually are.

In this paper, in contrast, we are mainly focused on situations where we already \emph{know} the lower bound holds (and can prove that in multiple ways), but we are striving to prove the known lower bound in a more constructive, algorithmic way. This sort of situation comes up routinely in applications of the probabilistic method, where an object we want can be constructed with randomness, but it is a major open problem to construct it deterministically. Our results indicate that there is a deep technical gap between the major complexity class separation problems, versus many lower bounds we know how to prove. The former type of lower bound problem automatically has constructive aspects built into it, while the latter type of lower bound requires a breakthrough in derandomization in order to be made constructive.

\section{Constructive Separations for Streaming and Query Algorithms imply Breakthrough Lower Bounds}
\label{sec:refuter-for-streaming-algorithms}

Streaming lower bounds and query complexity lower bounds are often regarded as well-understood, and certain lower bounds against one-tape Turing machines have been known for 50 years. In this section we show that surprisingly, making these separations constructive would imply breakthrough separations such as $\EXP^\NP \ne \BPP$ or even $\PTIME \ne \NP$.

\subsection{Making Most Streaming Lower Bounds Constructive Implies Breakthrough Separations}

We show that if randomized streaming lower bounds for \emph{any} language $L$ in $\NP$ can be made constructive, even with a $\PTIME^\NP$ refuter, then $\EXP^{\NP} \ne \BPP$.

\theoPNP*

% \begin{theorem}[Theorem~\ref{theo:PNP-constructive-to-EXPNP-lowb} restated]
% 	Let $f(n) \ge \omega(1)$. For every language $L \in \NP$, a $\PTIME^\NP$-constructive separation of $L$ from uniform randomized streaming algorithms with $O(n \cdot (\log n)^{f(n)})$ time and $O(\log n)^{f(n)}$ space implies $\EXP^{\NP} \neq \BPP$.
% \end{theorem}
\begin{remark}
    Let $V(x,y)$ be a verifier for $L$, and assume that the witness length $|y|$ is at most~$|x|$.\footnote{That is, $x \in L$ if and only if there exists $y \in \{0,1\}^*$ such that $|y| \le |x|$ and $V(x,y) = 1$.} Then the randomized streaming algorithms $A$ considered in Theorem~\ref{theo:PNP-constructive-to-EXPNP-lowb} can be further assumed to solve the search-version of $L$ with one-sided error in the following sense: (1) $A$ is also required to output a witness $y$ when it decides $x \in L$ (2) whenever $A$ outputs a witness~$y$, we have $V(x,y)=1$. 
\end{remark}

We need the following lemma for solving $\sMCSP$, which adapts an oracle algorithm from~\cite{McKayMW19}. The original algorithm of~\cite{McKayMW19} has two-sided error: that is, when $x \notin \MCSP[s(n)]$, there is a small probability that the algorithm outputs an incorrect circuit. We modify their approach with a carefully designed checking approach so that the algorithm has only one-sided error.

\begin{lemma}[{\cite[Theorem~1.2]{McKayMW19}}, adapted]\label{lm:McKayMW19-search-MCSP}
	Assuming $\NP \subseteq \BPP$, for a time-constructive $s \colon \mathbb{N} \to \mathbb{N}$, there is a randomized streaming algorithm for $\sMCSP[s(n)]$ on $N$-bit instances (where $N=2^n$) with $O(N \cdot s(n)^c)$ time and $O(s(n)^c)$ space for a constant $c$ such that the following holds.
	\begin{itemize}
		\item If the input $x \in \MCSP[s(n)]$, the algorithm outputs a circuit $C$ of size at most $s(n)$ computing $x$ with probability at least $1 - 1/N$.
		\item If the input $x \notin \MCSP[s(n)]$, the algorithm always outputs NO.
	\end{itemize}

	Alternatively, if we assume $\NP = \PTIME$ instead, the above randomized streaming algorithm can be made deterministic.
\end{lemma}
\begin{proof}
	We first recall the $\Sigma_3 \PTIME$ problem Circuit-Min-Merge introduced in~\cite{McKayMW19}; here, we will only consider the version with two given input circuits. In the following we identify the integer~$i$ from $[2^n]$ with the $i$-th string from $\{0,1\}^n$ (ordered lexicographically).
	
	\begin{figure}[ht]
	\renewcommand{\figurename}{Problem} 
	 \caption{Circuit-Min-Merge}
	 \label{algo:merge}
	 \centering
		$\text{Circuit-Min-Merge}[s(n)]$
	\flushleft	
		\textbf{Input:} Given two circuits $C_1,C_2$ on $n = \log N$ input bits and three integers $\alpha < \beta \le \gamma \in [2^n]$.
		
		\textbf{Output:} The lexicographically first circuit $C'$ of size at most $s(n)$ such that for all $\alpha \le z \le \beta-1$, $C'(z) = C_1(z)$, and for all $\beta \le z \le \gamma$, $C'(z) = C_2(z)$. If there are no such circuits, it outputs an all-zero string.
	\end{figure}

	Note that since $\NP \subseteq \BPP$, it follows that $\text{Circuit-Min-Merge}$ is also in $\BPP$. We can without loss of generality assume we have a $\BPP$ algorithm for $\text{Circuit-Min-Merge}$ with error at most $1/N^3$.

We give a brief overview of the proof idea. In the proof of the original two-sided error version \cite[Theorem 1.2]{McKayMW19}, they designed a streaming algorithm which maintains a circuit $C$ such that $C(z) = x_z$ for all the processed input bits $x_z$ so far, where $C$ is periodically updated as more input bits arrive, with the help of the $\BPP$ algorithm for Circuit-Min-Merge. 
 In our one-sided error case, we need to verify that the circuit returned by the $\BPP$ algorithm is indeed correct. In order to perform this verification efficiently, our streaming algorithm proceeds in a binary-tree-like structure (in contrast to the linear structure in \cite[Theorem 1.2]{McKayMW19}), so that we can reduce the total time spent on verification by performing expensive checks less frequently.

	After processing the first $p \in [2^n]$ bits of the input $x$, our streaming algorithm maintains a list of at most $n$ circuits. Specifically, let $p = \sum_{k=0}^{n} a_k \cdot 2^{k}$ be the binary representation of~$p$. For each $k \in [n]$, we maintain a circuit $C_k$ that is intended to satisfy $C_k(z) = x_z$ for all $\sum_{\ell > k} a_\ell \cdot 2^{\ell} < z \le \sum_{\ell \ge k} a_\ell \cdot 2^{\ell}$. Note that when $a_k = 0$, there is indeed no requirement on the circuit $C_k$ and we can simply set it to a trivial circuit.
	
	Now, suppose we get the $p + 1$ bit of the input $x$. We update the circuit list via the following algorithm.
	
	\begin{itemize}
		\item We initialize $D$ to be the linear-size circuit which outputs $x_{p+1}$ on the input $p+1$, and outputs $0$ on all other inputs.
		
		\item For $k $ from $0$ to $n$: %
		\begin{itemize}
			\item If $a_k = 1$, we set $D = \text{Circuit-Min-Merge}(C_k,D,\alpha,\beta,\gamma)$ with suitable $\alpha,\beta,\gamma$, and set $a_k = 0$ and $C_k$ to be a trivial circuit. We next check whether $D$ is indeed the correct output of $\text{Circuit-Min-Merge}(C_k,D,\alpha,\beta,\gamma)$ by going through all inputs in $[\alpha,\gamma]$. We output NO and halt the algorithm immediately if we found $D$ is not the correct output (if $\text{Circuit-Min-Merge}(C_k,D,\alpha,\beta,\gamma)$ outputs the all-zero string, we also output NO and halt the algorithm).
			
			\item If $a_k = 0$, we set $C_k = D$, and set $a_k\gets 1, a_{k-1},a_{k-2},\dots,a_0 \gets 0$ (in this way the binary counter $\sum_{k=0}^n a_k 2^k$ is incremented by $1$), and halt the update procedure.
		\end{itemize}
	\end{itemize}

	After we have processed the $2^n$-bit of $x$, we simply output $C_n$. If $x \in \MCSP[s(n)]$, then by a simple union bound, with probability at least $1-1/N$, all calls to our $\BPP$ algorithm for $\text{Circuit-Min-Merge}$ are answered correctly. In this case $C_n$ is a correct algorithm computing the input $x$. If $x \notin \MCSP[s(n)]$, since we have indeed checked the output of all $\text{Circuit-Min-Merge}$ calls, our algorithm will only output the circuit $C_n$ if it is indeed of size at most $s(n)$ and computes~$x$ exactly. Since $x \notin \MCSP[s(n)]$ implies there is no such circuit $C_n$, our algorithm always outputs NO in this case.
	
	For the running time, note that the above algorithm calls $\text{Circuit-Min-Merge}$ at most $N \cdot \log N \le O(N\cdot s(n))$ times on input of length $\widetilde{O}(s(n))$. Therefore calling $\text{Circuit-Min-Merge}$ only takes $N \cdot \poly(s(n))$ time in total. Note that merging $C_k$ and $D$ takes $2^k \cdot \poly(s(n))$ time to verify the resulting circuit, but this only happens at most $N/2^{k}$ times. So the entire algorithm runs in $N \cdot \poly(s(n))$ time and $\poly(s(n))$ space as stated.
\end{proof}

Now we are ready to prove Theorem~\ref{theo:PNP-constructive-to-EXPNP-lowb}.

\begin{proofof}{Theorem~\ref{theo:PNP-constructive-to-EXPNP-lowb}} The idea is to show that if $\EXP^{\NP} = \BPP$ then we can construct a randomized streaming algorithm for $L \in \NP$ that ``fools'' all possible $\PTIME^{\NP}$ refuters. Interestingly, the assumption is used in three different ways: (1) to bound the circuit complexity of the outputs of $\PTIME^{\NP}$ algorithms, (2) to obtain a randomized streaming algorithm that finds a small circuit encoding the input, and (3) to get an efficient algorithm to find a small circuit encoding a correct witness when it exists.

Let $L \in \NP$, and $V(x,y)$ be a polynomial-time verifier for $L$. Assuming $\EXP^{\NP} = \BPP$, we are going to construct a randomized streaming algorithm $A$, such that it solves $L$ correctly on all possible instances which can be generated by a $\PTIME^\NP$ refuter. 

Let $B$ be an arbitrary $\PTIME^\NP$ refuter. First, by Corollary~\ref{cor:small-circuit-special-case}, $\EXP^{\NP}= \BPP \subset \Ppoly$ implies that for all $n \in \mathbb{N}$, the length-$n$ string $B(1^n)$ has a circuit complexity of $w(n) = \polylog(n)$.

Second, note that $\EXP^{\NP} = \BPP$ also implies that $\NP \subseteq \BPP$. Let $f(n) \ge \omega(1)$ be time-constructive and for $n=2^m$ let $s(m) = (\log n)^{f(n) / c_1}$ for a sufficiently large constant $c_1 > 1$. By Lemma~\ref{lm:McKayMW19-search-MCSP}, we have a one-sided error randomized streaming algorithm $A_{\MCSP}$ for $\sMCSP[s(m)]$ with running time $n \cdot s(m)^{O(1)}$ and space $s(m)^{O(1)}$. Since $w(n) \le s(m)$, we apply $A_{\MCSP}$ to find an $s(m)$-size circuit $C$ encoding $B(1^n)$.

Now, we have an $s(m)$-size circuit encoding the $n$-bit input $B(1^n)$, and we wish to solve the Succinct-$L$ problem\footnote{Here, we define ``Succinct-$L$'' to be: given a circuit $C$ with $\ell$ input bits, decide whether $\tt(C) \in L$, where $\tt(C)$ is the truth table of $C$.} on this circuit. Note that Succinct-$L$ is a problem in $\NEXP$.

$\EXP^{\NP}=\BPP$ implies $\NEXP \subset \Ppoly$, so every Succinct-$L$ instance has a succinct witness with respect to the verifier $V$: this follows from the easy witness lemma of~\cite{ImpagliazzoKW02}. Formally, there exists a universal constant $k \in \N$ such that, for every $s(m)$-size circuit $D$ such that $\tt(D) \in L$, there exists an $s(m)^k$-size circuit $E$ such that $V(\tt(D),\tt(E)) = 1$.

We consider the following problem: 

\begin{quote}
Given an $s(m)$-size circuit $D$ with truth-table length $n=2^m$ and an integer $i \in [\log ( s(m)^k )]$, exhaustively try all circuits of size at most $s(m)^k$, find the first circuit $E$ such that $V(\tt(D),\tt(E)) = 1$, and output the $i$-th bit of the description of $E$.
\end{quote}

Note that the above algorithm runs in $2^{\poly(s(m))}$-time on $\poly(s(m))$-bit inputs, hence it is in $\EXP$. Since $\EXP = \BPP$, this problem is also in $\BPP$. Therefore there is a $\BPP$ algorithm which, given a Succinct-$L$ instance $D$ of size $s(m)$, outputs a description of a canonical circuit of size $s(m)^k$ which encodes a witness for input $\tt(D)$ with respect to verifier $V$.

Thus we obtain a randomized algorithm for $L$ on all instances with $s(m)$-size circuits. When the witness for $x$ has length at most $|x| = n$, the algorithm can take $n \cdot \poly(s(m))$ time to output the found witness, by outputting the truth-table of the circuit encoding the witness.

Setting $c_1$ to be large enough and putting everything together, we get the desired randomized streaming algorithm which solves all instances generated by $\PTIME^\NP$ refuters, which is a contradiction to our assumption. Therefore, it follows that $\EXP^\NP \ne \BPP$.
\end{proofof}

\subsection{Separating P and NP via Uniform-AC0-Constructive Separations}

Now we discuss a different setting, in which the existence of particular refuters would even imply $\PTIME \ne \NP$.

It is well-known (via communication complexity arguments) that $\DISJ$ does not have efficient streaming algorithms; in fact, any streaming algorithm must give incorrect answers on many inputs. So it is clear that counterexamples to $\DISJ$ exist, for every candidate streaming algorithm. But how efficiently can they be constructed? We show that the ability to construct counterexamples in uniform $\AC^0$ would actually imply $\PTIME \ne \NP$.

\theouniformAC*

% \begin{theorem}[Theorem~\ref{theo:uniformAC0-constructive-to-PNP-lowb} restated]
% 	Let $f(n)\ge\omega(1)$. A polylogtime-uniform-$\AC^0$-constructive separation of $\DISJ$ from uniform randomized streaming algorithm with $O(n \cdot (\log n)^{f(n)})$ time and $O(\log n)^{f(n)}$ space implies $\PTIME \neq \NP$.
% \end{theorem}

\begin{proof} 
We prove the contrapositive. Assuming $\PTIME = \NP$, we will show that there is an efficient streaming algorithm that solves all disjointness instances that are generated by polylogtime-uniform $\AC^0$ circuit families. 

From \cref{lem:pnpac0}, we know that the output string of any polylogtime-uniform $\AC^0$ circuit family has circuit size complexity at most $c(\log n)^c$ for some constant $c$.

Next, by Lemma~\ref{lm:McKayMW19-search-MCSP} we know that $\PTIME = \NP$ implies that $\sMCSP$ on input strings with circuits of size $c (\log n)^c$ can be solved by a streaming algorithm in $n\cdot (\log n)^{kc}$ time and $O(\log n)^{kc}$ space for some $k$. Also assuming $\PTIME=\NP$, $\DISJ$ on any $n$-bit input represented by a $c (\log n)^c$-size circuit can be solved in $ck (\log n)^{ck}$ time for some $k$; indeed, the ``Succinct-$\DISJ$'' problem \emph{given a circuit $C$ on $n+1$ inputs, does its truth table on $2^{n+1}$ inputs encode two $2^n$-bit strings which are disjoint?} is a $\co\NP$ problem.

For every function $f(n)\ge\omega(1)$, we can therefore design a streaming algorithm for $\DISJ$ as follows. First, on an input $x$, the algorithm solves $\sMCSP$ using $n \cdot (\log n)^{f(n)}$ time and $(\log n)^{f(n)}$ space to get an $O((\log n)^{f(n)})$ size circuit $C$ encoding $x$ (we abort the algorithm if it ever uses more than this time complexity or space complexity).
Then, we run a $(\log n)^{O(f(n))}$-time algorithm for Succinct-$\DISJ$ on the circuit $C$ (in the theorem statement we omit the big-O on the exponent because we can use $f(n)/c'$ instead of $f(n)$ for a sufficiently large constant $c'$). This will correctly decide disjointness on all inputs $x$ that are generated by a polylogtime-uniform $\AC^0$ circuit family. 
\end{proof}
	 
\subsubsection{Constructive Separations in Query Complexity}
	Finally we show certain uniform-$\AC^0$-constructive separations in query complexity would imply $\PTIME \ne \NP$. 
	
	\theorefuterqc*
	
	% \begin{theorem}[Theorem~\ref{theo:refuter-qc-to-PNP} restated]
	% Let $\eps$ be a function of $n$ satisfying $\eps \leq 1/(\log n)^{\omega(1)}$, and $1/\eps$ is a positive integer computable in $\poly(1/\eps)$ time given $n$ in binary.
	% \begin{itemize}
	%     \item If there is a polylogtime-uniform-$\AC^0$-constructive separation of $\GapMAJ_{n,\eps}$ from %
	%     randomized query algorithms $A$ using $o(1/\eps^2)$ queries and $\poly(1/\eps)$ time, then $\NP \ne \PTIME$.
	%     \item If there is a polylogtime-uniform-$\NC^1$-constructive separation of $\GapMAJ_{n,\eps}$ from %
	%     randomized query algorithms $A$ using $o(1/\eps^2)$ queries and $\poly(1/\eps)$ time, then $\PSPACE \ne \PTIME$.
	% \end{itemize}
	% \end{theorem}
	\begin{proof}
Assuming $\PTIME = \NP$, we will show that there is an efficient query algorithm that solves all $\GapMAJ_{n,\eps}$ instances that are generated by polylogtime-uniform $\AC^0$ circuit families. 

At the beginning, our query algorithm first computes the value of $\eps$ in $\poly(1/\eps)$ time.

From \cref{lem:pnpac0},	if $\PTIME=\NP$, then for every polylogtime-uniform $\AC^0$ circuit family $\{C_n\}$, the $n$-bit output of $C_n(1^n)$ has circuit size $(c\log n)^c$ for some constant $c$. (The same size bound also holds for polylogtime-uniform $\NC^1$ circuits, under the stronger assumption $\PTIME=\PSPACE$.)
By the assumption that $\eps = \eps(n) \le 1/(\log n)^{\omega(1)}$, this circuit size is at most $(c\log n)^c \le 1/\eps^{0.9}$ for sufficiently large $n$, and the number of circuits of size at most $1/\eps^{0.9}$ is $2^{O(\eps^{-0.9} \log \eps^{-1})}$. Hence,
such a circuit can be PAC-learned with error $\eps/2$ and failure probability $\delta=1/10$ using $O(\eps^{-1}\cdot (\eps^{-0.9} \log \eps^{-1} + \log \frac{1}{\delta})) \le O(\eps^{-1.91})$ samples (random queries) (see e.g.\ \cite[Theorem 2.5]{mohri2018foundations} on learning a finite class of functions). The learning algorithm achieving this sample complexity simply computes a minimum-size circuit that is consistent with all the observed samples. 
	
	Under the assumption of $\PTIME = \NP$, this learning algorithm can be executed in $\poly(1/\eps)$ time. Indeed, the following problem is in the polynomial-time hierarchy:
	\begin{quote}
	\emph{Given a set $L=\{(x_i,y_i)\} \subset \{0,1\}^n \times \{0,1\}$, a positive integer $s$, and an index $j$, output the $j$-th bit of the lexicographically first circuit $C$ of size at most $s$ such that $C(x_i)=y_i$ for all $i$.}
	\end{quote}
	Assuming $\PTIME = \NP$ (or $\PTIME = \PSPACE$) the above problem is in $\PTIME$ and hence can be solved in $\poly(|L|,s,\log n)$ time. Here $s = 1/\eps^{0.9}, |L|\le O(\eps^{-1.91})$, so we can find a minimum-size circuit consistent with any given input/output sample in $\poly(1/\eps,\log n) = \poly(1/\eps)$ time. Let $D$ be the circuit we have learned.
	
	Next, we decide $\GapMAJ_{n,\eps/2}$ on the truth table of $D$, by computing its average output value on $\Theta(1/\eps^2)$ uniform random inputs. This process takes $\poly(1/\eps)$ time, and makes no queries to the original input string. 
	Since the learned circuit $D$ only has error $\eps/2$ compared with the original input string, we can simply return the result as our answer to the original $\GapMAJ_{n,\eps}$ problem.
 The overall algorithm has success probability $2/3$, time complexity at most $\poly(1/\eps)$, and sample complexity $O(\eps^{-1.91}) = o(1/\eps^{2})$, because we do not need further samples from the original input string after we already learned the circuit $D$. 
	\end{proof}

	\subsection{Constructive Separations for One-Tape Turing Machines imply Breakthrough Lower Bounds}

	Maass~\cite{pal} showed that a one-tape nondeterministic Turing machine takes at least $\Omega(n^2)$ time to decide the language of palindromes $\PAL=\{x_n\cdots x_1x_1\cdots x_n \mid x_1,\dots,x_n \in \{0,1\}^n$, $n \in \N\}$. This is a very basic lower bound that is often cited as a canonical application of communication complexity. In this subsection, we show that a constructive proof of this lower bound would imply a breakthrough circuit lower bound. 
	
	In fact, we will prove a much more general statement. We will also generalize the proof to show that for every language $L$ computable by nondeterministic $n^{1 + o(1)}$-time RAMs, a constructive proof that ``$L$ cannot be decided by $n^{1.1}$-size nondeterministic one-tape Turing machines'' would yield uniformly-computable functions with exponential circuit complexity. That is, we would obtain major circuit lower bounds even from the task of distinguishing RAMs from one-tape Turing machines in a constructive way.
	
	We begin by a simple lemma showing  that nondeterministic one-tape Turing machines can solve $\PAL$ on inputs that have small circuits.
	
	\begin{lemma} \label{lemma:fast-NBP}
	For every constant $\delta \in (0,1]$, there is  a nondeterministic $n^{1+O(\delta)}$-time one-tape Turing machines solving $\PAL$ on every $x$ with circuit complexity at most $|x|^{\delta}$.
	\end{lemma}
	\begin{proof}
	Let $\delta \in (0,1]$.	Our nondeterministic (one-tape) Turing Machine $M$ runs as follows: 
	\begin{quote}
	$M$ guesses a circuit $C$ of size $n^\delta$, and checks that $C(i)$ equals the $i$-th input bit for all $1\le i\le n$, which can be done in $n\cdot n^{O(\delta)}$ time by moving the head on the tape from the first input bit to the last, while storing the $n^\delta$-size circuit $C$ in the cells close to the current position of the head.  Finally $M$ checks that the string $C(1)C(2)\cdots C(n)$ is a palindrome by evaluating $C$ on every $i$ and $n-i$, in $n\cdot n^{O(\delta)}$ total time. $M$ accepts the input on a guess $C$ if and only if all checks are passed.
	\end{quote}
	Observe that $M$ recognizes $\PAL$ correctly on every string $x$ with circuit complexity at most $n^{\delta}$, and its running time is bounded by $n^{1 + O(\delta)}$.
	\end{proof}
	
	Now we show that breakthrough separations follow from constructive proofs of lower bounds for $\PAL$.
	
	\begin{theorem} \label{thm:pal}
	The following hold:
	\begin{itemize}
	    \item A $\PTIME^{\NP}$-constructive separation of $\PAL$ from nondeterministic $O(n^{1.1})$ time one-tape Turing machines implies $\E^{\NP} \not \subset \SIZE[2^{\delta n}]$ for some constant $\delta >0$.
	    
        \item A $\PTIME$-constructive separation of $\PAL$ from nondeterministic $O(n^{1.1})$ time one-tape Turing machines implies $\E \not \subset \SIZE[2^{\delta n}]$ for some constant $\delta >0$.
        
        \item A LOGSPACE-constructive separation of $\PAL$ from nondeterministic $O(n^{1.1})$ time one-tape Turing machines implies $\PSPACE \not \subset \SIZE[2^{\delta n}]$ for some constant $\delta >0$.
	\end{itemize}
	\end{theorem}
	\begin{proof}
	We will only prove the first item; it is straightforward to generalize to the other two items. Let $\delta > 0$ be a small enough constant such that, by Lemma~\ref{lemma:fast-NBP}, there is a nondeterministic $O(n^{1.1})$-time one-tape Turing machine solving $\PAL$ correctly on inputs $x$ with circuit complexity at most $n^\delta$.
	
	Now suppose there is a $\PTIME^{\NP}$ refuter for $M$: a polynomial-time algorithm $A$ with an $\NP$ oracle, which on input $1^n$ outputs an $n$-bit string. Assuming that $\E^{\NP} \subset \SIZE[2^{\delta_1 n}]$ for a constant $\delta_1 > 0$ that is small enough compared to $\delta$, by Lemma~\ref{lemma:P-NP-small-circuit}	there is a circuit $C$ of size at most $n^{O(\delta_1)} \le n^{\delta}$ that on input $(n,i)$ computes the $i$-th bit of $A(1^n)$. 
	That is, the output of any such $A$ on $1^n$ has circuit complexity at most $n^{\delta}$. By construction, $M$ will always decide $A(1^n)$ correctly, contradicting the assumption that $A$ is a refuter.
	Hence, there must exist a constant $\delta > 0$ such that $\E^{\NP} \not \subset \SIZE[2^{\delta n}]$.
\end{proof}

We say a family of $3$-$\SAT$ formulas $\{C_n\}_{n \in \N}$ such that $C_n$ has $S(n)$ clauses is \emph{strongly explicit}, if there is an algorithm $A$ such that $A(n,i)$ outputs the $i$-th clause of $C_n$ in $\polylog(S(n))$ time. We need the following efficient reduction from nondeterministic $T(n)$-time RAMs to $T(n) \cdot \polylog(T(n))$-size $3$-$\SAT$ instances. 

\begin{lemma}[\cite{Tourlakis01,FortnowLMV05}]\label{lemma:red}
Let $M$ be a $T(n)$-time nondeterministic RAM. There exists a strongly explicit family of $3$-SAT formulas $\{C_n\}_{n \in \N}$ of $T \cdot \polylog(T)$ size, such that for every $x \in \{0,1\}^n$, $M(x) = 1$ if and only if there exists $y$ such that $C_n(x,y) = 1$.
\end{lemma}

Now we are ready to generalize Theorem~\ref{thm:pal} to other problems.

\theogeneral*

% \begin{theorem}[\Cref{thm:general} restated]
%    For every language $L$ computable by a nondeterministic $n^{1+o(1)}$-time RAM, a $\PTIME^{\NP}$-constructive separation of $L$ from nondeterministic $O(n^{1.1})$-time one-tape Turing machines implies $\E^{\NP} \not \subset \SIZE[2^{\delta n}]$ for some constant $\delta >0$.
% \end{theorem}

\begin{proof}
    Let $M_{\sf RAM}$ be a nondeterministic $n^{1+o(1)}$-time RAM for $L$. We apply Lemma~\ref{lemma:red} to obtain a strongly explicit family of $3$-SAT formulas $\{C_n\}_{n \in \N}$ with $n^{1+o(1)}$ size and $s = n^{1+o(1)}$ variables.

    Let $\delta_1 > 0$ be a small enough constant, and consider the following nondeterministic (one-tape) Turing machine $M$:
    \begin{quote}
	$M$ guesses a circuit $D$ of size $n^{\delta_1}$, and checks that $D(i)$ equals the $i$-th input bit for all $1\le i\le n$, which can be done in $n\cdot n^{O(\delta_1)}$ time by moving the head on the tape from the first input bit to the last, while storing the $n^{\delta_1}$-size circuit $D$ in the cells close to the current position of the head.
	
	Next, $M$ guesses a circuit $E$ of size $n^{\delta_1}$, and accepts if and only if \[D(1),\dotsc,D(n),E(1),\dotsc,E(s-n)\] satisfies $C_n$. Note that this can be checked in $n^{1+O(\delta_1)}$ time by enumerating all $n^{1+o(1)}$ clauses in~$C_n$ and evaluating $D$ and $E$ to obtain the assignments to the corresponding variables.
	\end{quote}
	
	We take $\delta_1$ to be small enough so that the above machine $M$ runs in $O(n^{1.1})$ time. Suppose there is a $\PTIME^\NP$ refuter $B$ for $L$ against $M$, and we further assume towards a contradiction that $\ETIME^\NP \subset \SIZE(2^{\delta n})$ for all $\delta >0$.
	
	By Lemma~\ref{lemma:P-NP-small-circuit}, it follows that $B(1^n)$ has an $n^{\delta_1}$-size circuit. It also follows that if $B(1^n) \in L$, then the lexicographically first string $y_n \in \{0,1\}^{s-n} $ such that $C_n(B(1^n),y_n)$ has an $n^{\delta_1}$-size circuit. By Lemma~\ref{lemma:red}, this means that $M$ solves $B(1^n)$ correctly, a contradiction. Hence, we have that $\ETIME^\NP \not\subset \SIZE(2^{\delta n})$ for some $\delta >0$.\footnote{As noted by an anonymous reviewer, if in the statement of~\autoref{thm:general} we instead have a $\PTIME$-constructive separation, then we would get $\NEXP \not \subset\PTIME/\poly$. The idea is that one can use both $\NEXP \subset \PTIME/\poly$ and the easy witness lemma~\cite{ImpagliazzoKW02} to argue that there exists a string $y_n$ of small circuit complexity such that $C_n(B(1^n),y_n)$ holds.}
\end{proof}

We conclude this section with a remark on the proofs. In the proofs of Lemma~\ref{lemma:fast-NBP} and Theorem~\ref{thm:general}, we can naturally view our constructions as \emph{nondeterministic streaming algorithms} with total time $n^{1+O(\delta)}$ and space $n^{O(\delta)}$. Hence, both results apply to low-space nondeterministic algorithms equally well. We only state the generalization of Theorem~\ref{thm:general} below.

\begin{remark}
	For every language $L$ computable by a nondeterministic $n^{1+o(1)}$-time RAM, a $\PTIME^{\NP}$-constructive separation of $L$ from nondeterministic $O(n^{1.1})$-time $n^{0.1}$-space streaming algorithms implies $\E^{\NP} \not \subset \SIZE[2^{\delta n}]$ for some constant $\delta >0$.
\end{remark}

This remark is stronger than Theorem~\ref{thm:general}, as any $(n\cdot t)$-time $t$-space nondeterministic streaming algorithm can be simulated by an $n \cdot \poly(t)$ time nondeterministic one-tape Turing machine (see, e.g.,~\cite[Lemma~9]{CheraghchiHMY21}). However, we have chosen not to emphasize it because the model of ``nondeterministic streaming'' is less common.

\section{Constructive Separations for MCSP Imply Breakthrough Lower Bounds}\label{sec:constructive-sep-MCSP}

In this section we show that constructive separations for $\MCSP$ against uniform $\AC^0$ imply breakthrough lower bounds. In particular, we prove Theorem~\ref{thm:mag_refuter} (restated below for convenience).
Recall that a circuit of size $S$ is said to be polylogtime-uniform, if there is a $\polylog(S)$-time algorithm that decides the type of a gate $g$ given its $O(\log S)$-bit index,  and decides whether there is a wire from gate $g_1$ to gate $g_2$ given their indices.

\theomag*

% \begin{theorem}[Theorem~\ref{thm:mag_refuter} restated]
%     Let $s(n) \ge n^{\log(n)^{\omega(1)}}$ be any time-constructive super-quasipolynomial function.
%     In the following, we consider $\MCSP[s(n)]$ and Parity problems of input length $N=2^n$.
%     The following hold:
%     \begin{enumerate}[label=(\arabic*)]
%         \item (Major Separation from Constructive Lower Bound) If there is a polylogtime-uniform $\AC^0[\quasipoly]$ refuter for $\MCSP[s(n)]$ against every polylogtime-uniform $\AC^0$ algorithm, then $\PTIME \ne \NP$.
%         \item (Constructive Lower Bound Should Exist) If $\PH \not \subseteq \SIZE(s(n)^2)$, then there is a polylogtime-uniform-$\AC^0[\quasipoly]$ refuter for $\MCSP[s(n)]$ against every polylogtime-uniform $\AC^0$ algorithm.
%         \item (Somewhat Constructive Lower Bound) For $s(n) \le o(2^n/n)$, there is a polylogtime-uniform-$\AC^0[2^{\poly(s(n))}]$ refuter for $\MCSP[s(n)]$ against every polylogtime-uniform $\AC^0$ algorithm.
%         \item (Constructive Lower Bound for a Different Hard Language) There is a  polylogtime-uniform-$\AC^0[\quasipoly]$-list-refuter for $\Parity$ against every polylogtime-uniform $\AC^0$ algorithm.
%     \end{enumerate}
% \end{theorem}
Throughout this section, we use $N$ to refer to the size of a truth table of a Boolean function on $n = \log(N)$ bits.

To prove Theorem~\ref{thm:mag_refuter}, we will heavily use known results about pseudo-random generators against $\AC^0$.

\begin{theorem} [\cite{Nisan91, Viola05}]\label{thm:Nisan}  Let $d,c$ be any positive integers. There is a pseudo-random generator $G =\{G_N\}, G_N\colon \{0,1\}^{\log(N)^{O(d)}} \rightarrow \{0,1\}^N$, such that for each $N$, the PRG $G$ $1/N$-fools depth-$d$ $\AC^0$ circuits of size $N^c$. 
Moreover, $G$ is computable by polylogtime-uniform-$\AC^0$ circuits of size $\poly(N)$, and $G_N(z)$ has circuit complexity $\polylog(N)$ for each seed $z$ of length $\log(N)^{O(d)}$.
\end{theorem}

\begin{corollary}[\cite{Allender01}] \label{cor:MCSPhardAC0}  
For $N=2^n$, let $s(n) \ge n^{\omega(1)}$ be any time-constructive function such that $s(n) \le o(2^n/n)$. Then $\MCSP[s(n)]$ is not in $\AC^0$.
\end{corollary}

Corollary \ref{cor:MCSPhardAC0} follows from Theorem \ref{thm:Nisan} by observing that $\MCSP[s(n)]$ distinguishes the uniform distribution on $N=2^n$ bits from the output of $G_N$, since every output of $G_N$ is a YES instance of $\MCSP[s(n)]$, while a random string of length $N$ is a NO instance with high probability. In fact, it follows that for $s(n)$ quite close to maximum, the $\AC^0$ lower bounds are exponential (but with an inverse dependence in the exponent on the circuit depth), similar to known lower bounds for $\Parity$.

First, we show that uniform $\AC^0$ refuters for separations of $\MCSP$ from uniform $\AC^0$ would solve the main open problem in complexity theory. This establishes the first item of Theorem~\ref{thm:mag_refuter}. We find it more convenient here to state the size bound $s(n)$ for $\MCSP$ in terms of the input size $N = 2^n$ than in terms of $n$, so we usually write $\MCSP[f(N)]$ (where $f(N)=s(n)$). Since $s(n)$ is required to be a time-constructive function in the statement of \cref{thm:mag_refuter}, $f(N)$ should be computable in $\poly(f(N))$ time given $N$ represented in binary.

The following theorem implies Item~(1) of \cref{thm:mag_refuter}, since the constant $1$ function can be trivially implemented by a polylogtime-uniform $\AC^0$ algorithm.
\begin{proposition}[Item~(1) of Theorem~\ref{thm:mag_refuter}]\label{thm:MCSPConstructivePNP}
	Let $f(N) \ge 2^{\log \log(N)^{\omega(1)}}$ be 
 a function computable in $\poly(f(N))$ time.
 If there exists a polylogtime-uniform-$\AC^0[\quasipoly]$ refuter for $\MCSP[f(N)]$ against 
 the constant $1$ function,
  then $\PTIME \ne \NP$.
\end{proposition}

\begin{proof}
Assume that $\PTIME = \NP$ and that there is a polylogtime-uniform-$\AC^0$ refuter $R$ for $\MCSP[f(N)]$ against the constant $1$ function. We derive a contradiction. 
Using the same argument as in the proof of Theorem~\ref{theo:uniformAC0-constructive-to-PNP-lowb}, the refuter $R$ always outputs a string $x$ of circuit complexity $2^{\log \log(N)^{O(1)}}$. But such a string is a YES instance of $\MCSP[f(N)]$ since $f(N) \ge 2^{\log \log(N)^{\omega(1)}}$. This contradicts the assumption that $R$ refutes the algorithm that always outputs YES.
\end{proof}

By inspecting the proof carefully, it can be seen that the conclusion above holds even if the hypothesis is that there is a quasipolynomial-size uniform $\AC^0$ list-refuter running in quasi-polynomial time.

Next, we show that if a certain natural circuit lower bound assumption holds for the Polynomial Hierarchy, we do get the strongly constructive separations we seek. We obtain these separations by using a win-win argument: for any uniform $\AC^0$ algorithm, either the algorithm outputs NO with noticeable probability, in which case the refuter exploits a PRG whose range is supported on strings of low circuit complexity, or it outputs YES with noticeable probability, in which case the refuter exploits a PRG (obtained using our assumption) whose range is supported on strings of high circuit complexity. This establishes the second item of Theorem~\ref{thm:mag_refuter}.

\begin{proposition}[Item~(2) of Theorem~\ref{thm:mag_refuter}] \label{thm:MCSPConstructiveConditional}
      Let $f(N) \ge 2^{\log \log(N)^{\omega(1)}}$ be 
 a function computable in $\poly(f(N))$ time.
      If $\PH \not \subseteq \SIZE(f(N)^2)$, then there exists a polylogtime-uniform-$\AC^0[\quasipoly]$ refuter for $\MCSP[f(N)]$ against every polylogtime-uniform $\AC^0$ algorithm.
\end{proposition}

\begin{proof}
 Let $f$ be as in the statement of the theorem, $F \in \PH$ be such that $F \not \in \SIZE(f(N)^2)$, and $A$ be a polylogtime-uniform-$\AC^0$ algorithm. We construct a polylogtime-uniform-$\AC^0[\quasipoly]$ refuter $R$ against $A$.

Let $G$ be the PRG from Theorem~\ref{thm:Nisan} where $d$ is the depth of the uniform $\AC^0$ algorithm $A$, and let $G' =\{G'_N\}$ be the generator from $\log(N)^{O(d)}$ bits to $N$ bits defined by $G'_N(z) = G_N(z) \oplus y_N$ for each seed $z$, where $y_N$ is the truth table of $F$ on $\log(N)$ input bits (recall $N$ is a power of two). The refuter $R$ outputs the lexicographically first string $x$ in the range of $G$ such that $A(x) = 0$, or in case such a string does not exist, the lexicographically first string $x'$ in the range of $G'$ such that $A(x') = 1$. We will show that either $x$ or $x'$ exists. Note that $R$ can be implemented by polylogtime-uniform quasipolynomial-size $\AC^0$ circuits, since both $G$ and $G'$ have quasi-polynomial sized range and can be computed in uniform $\AC^0$ - this is true for $G$ by Theorem~\ref{thm:Nisan} and it is true for $G'$ because the truth table of any $\PH$ function on $\log(N)$ bits can be computed by uniform $\AC^0$ circuits of size $\poly(N)$.

Since $G$ $1/N$-fools depth-$d$ $\AC^0$ circuits and $G'$ is a linear translate of the range of $G$, $G'$ also $1/N$-fools depth-$d$ $\AC^0$ circuits. We show that there is either a string $x$ in the range of $G$ such that $A(x) = 0$ or a string $x'$ in the range of $G'$ such that $A(x') = 1$. $A$ either outputs NO with probability at least $1/2$ on randomly chosen input of length $N$, or it outputs YES with probability at least $1/2$. In the first case, since $G$ $1/N$-fools $A$, there is a string $x$ in the range of $G$ such that $A(x) = 0$. Moreover, since every string in the range of $G$ is a YES instance of $\MCSP[f(N)]$ by Theorem~\ref{thm:Nisan}, we have that $x$ refutes that $A$ solves $\MCSP[f(N)]$ correctly. In the second case, since $G'$ $1/N$-fools $A$, there is a string $x'$ in the range of $G'$ such that $A(x') = 1$. Moreover, since $F \not \in \SIZE(f(N)^2)$ and every string in the range of $G$ has $\polylog(N)$ size circuits, it follows that every string in the range of $G'$ is a NO instance of $\MCSP[f(N)]$. Thus $A$ makes a mistake on $x'$, implying that $R$ is a correct refuter.
\end{proof}

We also show that slightly weaker constructive separations than desired do hold unconditionally. The argument is similar to the argument in the proof of Theorem~\ref{thm:MCSPConstructiveConditional}, but since we do not use an assumption, we need to argue differently in the case where the algorithm we are refuting outputs YES with high probability. We do so by exploiting the sparsity of the language against which we are showing a lower bound. This establishes the third item of Theorem~\ref{thm:mag_refuter}.

\begin{proposition}[Item~(3) of Theorem~\ref{thm:mag_refuter}] \label{thm:MCSPConstructiveUnconditional}
      Let $f(N) \ge \log(N)^{\omega(1)}$ be
 a function computable in $\poly(f(N))$ time,
      such that $f(N) \le o(N/\log(N))$. There is a polylogtime-uniform-$\AC^0[2^{\poly(f(N))}]$ %
      refuter for $\MCSP[f(N)]$ against every polylogtime-uniform $\AC^0$ algorithm.
\end{proposition}

\begin{proof}
Given a polylogtime-uniform $\AC^0$ algorithm $A$, we define a uniform $\AC^0$ refuter $R$ of size $2^{\poly(f(N))}$. For any $d$, let $G^d$ be the PRG from Theorem~\ref{thm:Nisan} corresponding to depth $d$, and let $G = G^d$ where $d$ is the depth of the uniform $\AC^0$ algorithm $A$, so that $G$ $1/N$-fools $A$ on input length $N$. 
Let $G'$ be the generator with seed length $\poly(f(N))$ obtained by truncating the output of $G^{d'}_{2^{f(N)^c}}$ to $N$ bits (where $d'$ and $c$ are to be specified later), so that $G'$ $1/N$-fools depth-$d'$ $\AC^0$ circuits of $2^{\poly(f(N))}$ size.
$R$ works as follows. It outputs the lexicographically first string $x$ in the range of $G$ for which $A(x) = 0$, and if such an $x$ does not exist, it outputs the lexicographically first string $x'$ in the range of $G'$ that is not a YES instance of $\MCSP[f(N)]$ for which $A(x') = 1$. We show that such an $x'$ always exists in the case that $x$ does not, and that moroeover $A$ is a correct refuter.
Since $G$ and $G'$ can be computed by uniform $\AC^0$ circuits of size exponential in $\poly(f(N))$ and moreover the YES instances of $\MCSP[f(N)]$ can be enumerated by uniform $\AC^0$ circuits of size exponential in $\poly(f(N))$, we have that the refuter can be implemented by uniform $\AC^0$ circuits of size exponential in $\poly(f(N))$. %

Either $A$ outputs NO with probability greater than $1/2$ on a uniformly chosen input of length $N$, or it does not. In the first case, since $G$ $1/N$-fools $A$, there must be a string $x$ in the range of $G$ for which $A(x) = 0$. Moreover, since every string in the range of $G$ has circuit complexity $\polylog(N) \ll f(N)$, we have that $x$ is a YES instance of $\MCSP[f(N)]$, and hence the refuter correctly outputs an input on which $A$ makes a mistake in this case.

Suppose $A$ outputs YES with probability at least $1/2$. We define a uniform $\AC^0$ algorithm $A'$ of size $2^{\poly(f(N))}$ as follows. $A'$ first enumerates all YES instances of $\MCSP[f(N)]$. Note that there are at most $2^{O(f(N)\log N)}$ YES instances, and they can be enumerated by an $\AC^0$ algorithm of size $2^{\poly(f(N))}$ by running over all circuits of size at most $f(N)$ and guessing and checking their computations. $A'$ checks if its input $x'$ is in the list of YES instances of $\MCSP[f(N)]$ or not. If it is, it outputs NO, otherwise it runs $A$ on $x'$ and outputs the answer. $A'$ can be implemented by polylogtime-uniform $\AC^0$ circuits of size $2^{\poly(f(N))}$ and constant depth. 
Note that $A'$ outputs YES with probability at least $1/2 - \frac{2^{O(f(N)\log N)}}{2^N}>0.49$ (where we used $f(N)\le  o(N/\log N)$).
Now, by choosing the parameters $c$ and $d'$ in the first paragraph large enough so that $G'$ $1/N$-fools $A'$, we have that at least a $0.49-1/N$ fraction of outputs $x'$ of $G'$ have $A'(x') = 1$, and hence there is a lexicographically first such output. Moreover, since $A'$ outputs NO on all YES instances of $\MCSP[f(N)]$, it must be the case that $x'$ is a NO instance of $\MCSP[f(N)]$. By definition of $A'$ we know $A(x')=A'(x')=1$, so $A$ makes a mistake on $x'$ when trying to solve $\MCSP[f(N)]$.
\end{proof}

Finally, we observe that the strongly constructive separations we seek do hold in the case of the well-known lower bound for $\Parity$ against $\AC^0$. Indeed, in this case we actually get an oblivious list-refuter (a.k.a. an explicit obstruction), meaning that the list-refuter does not need to depend on the algorithm being refuted. This establishes the fourth item of Theorem~\ref{thm:mag_refuter}.

\begin{theorem}[\cite{Ajtai83, FSS84, Yao85, Hastad86}] \label{thm:hastad} 
For each integer $d$, $\Parity$ does not have depth-$(d+1)$ $\AC^0$ circuits of size $2^{O(N^{1/d})}$.
\end{theorem}

\begin{proposition}[Item~(4) of Theorem~\ref{thm:mag_refuter}]
   There is a  polylogtime-uniform-$\AC^0[\quasipoly]$-list-refuter for $\Parity$ against every polylogtime-uniform $\AC^0$ algorithm.
\end{proposition}
\begin{proof}
In fact, we show that for all $d$ there is an oblivious list-refuter $R$ that refutes depth-$(d+1)$ $\AC^0$ algorithms by outputting a $\quasipoly$-size set of strings of length $N$. The list-refuter $R$ simply outputs the set of all strings of the form $y0^{N-\log(N)^d}$ where $y \in \{0,1\}^{\log(N)^d}$. Suppose, for the sake of contradiction, that there is a uniform depth-$(d+1)$ $\AC^0$ algorithm $A$ that correctly solves $\Parity$ on all strings output by $R$. Then we can compute $\Parity$ by circuits of size $2^{O(m^{1/d})}$ on input $y$ of length $m$ as follows: pad $y$ to length $2^{m^{1/d}}$ by suffixing it with zeroes, then run $A$ on the padded string. This contradicts the lower bound of Theorem~\ref{thm:hastad}.
\end{proof}

\section{Most Conjectured Uniform Separations Can Be Made Constructive}\label{sec:uniform-can-constructive}

In this section we show many uniform separations imply corresponding refuters. We will prove Theorem~\ref{theo:refuter-uniform-sep} (restated below).

\renewcommand{\footnoterefuteruniformsep}{}
\theorefuteruniformsep*

% \begin{theorem}[\cref{theo:refuter-uniform-sep} restated]
% 	Let $\caC \in \{\PTIME, \ZPP,\BPP \}$ and let $\caD \in \{ \NP, \Sigma_2 \PTIME, \dotsc ,\Sigma_k \PTIME,$  $\dotsc,  \PP, \PSPACE, \EXP,\NEXP,\EXP^\NP\}$. Then $\caD \nsubseteq \caC$ implies that for every paddable $\caD$-complete language $L$, there is a $\caC$-constructive separation of $L \notin \caC$.\footnote{Throughout this paper when we say  a language $L$ is $\caD$-complete, we mean it is $\caD$-complete under polynomial-time many-one reductions. A language $L$ is \emph{paddable} if there is a deterministic polynomial-time algorithm that receives $(x,1^n)$ as input, where string $x$ has length at most $n-1$, and then outputs a string $y\in \{0,1\}^n$ such that $L(x)=L(y)$.}
	
% 	Furthermore, $\ParP \nsubseteq \caC$ implies that for every paddable $\ParP$-complete language $L$, there is a $\BPP$-constructive separation of $L \notin \caC$.
% \end{theorem}

We will prove the case of $\caD \in \{\PSPACE, \EXP, \NEXP, \EXP^\NP\}$ in \cref{sec:refuter-nexp}, $\caD\in \{\Sigma_k \PTIME\}_{k \ge 1}$ in \cref{sec:refuter-np}, and $\caD\in \{\PP, \ParP\}$ in \cref{sec:refuter-pp}.

In the proofs of this section we frequently use the following notion of list-refuters, which is a relaxation of refuters (\cref{defn:refuter}) that allows outputting a \emph{constant} number of strings with possibly \emph{different lengths} (instead of a single string), and only requires at least one of the strings is a counterexample: 
\begin{definition}[Constant-size list-refuters]
\label{defn:constsizelistref}
		For a  function $f \colon \{0,1\}^{\star} \rightarrow \{0,1\}$ and an algorithm $A$, a \emph{constant-size} $\PTIME$-\emph{list-refuter} for $f$ against $A$ is a deterministic polynomial time algorithm $R$ that, given input $1^n$, prints a list of $c$ strings $x_n^{(1)},x_n^{(2)},\dots,x_n^{(c)}\in \{0,1\}^*$ (for a constant~$c$ independent of $n$),  such that for infinitely many $n$,  there exists $i\in [c]$ for which $A(x_n^{(i)})\neq f(x_n^{(i)})$.
	Moreover, for every $i\in [c]$, there is a strictly increasing polynomial $\ell^{(i)}\colon \N \to \N$ such that $|x_n^{(i)}|=\ell^{(i)}(n) \ge n$ for all integers $n$.

  This definition can be generalized to constant-size $\BPP$-list-refuters and constant-size $\ZPP$-list-refuters similarly to \cref{defn:refuter}.

\end{definition}

The following simple lemma says that a constant-size list-refuter as defined in \cref{defn:constsizelistref} can be converted to a refuter as defined in \cref{defn:refuter}.
\begin{lemma}
	\label{rem:list}
		For function $f \colon \{0,1\}^{\star} \rightarrow \{0,1\}$ and algorithm $A$,
  \begin{itemize}
      \item 
  A constant-size $\PTIME$-list-refuter for $f$ against $A$ implies that there exists a $\PTIME$-refuter for $f$ against~$A$.
  \item 
	For $\caD \in \{\BPP,\ZPP\}$,	 a pseudo-deterministic constant-size $\caD$-list-refuter (i.e., for each input length $n$ there is a canonical list such that the refuter outputs the canonical list with $1-o(1)$ probability given input $1^n$) for $f$ against $A$ implies that there exists a $\caD$-refuter for $f$ against~$A$.
  \end{itemize}
  
\end{lemma}
\begin{proof}
Suppose we have a constant-size $\PTIME$-list-refuter (\cref{defn:constsizelistref}) which outputs the list $x_n^{(1)}$, $x_n^{(2)}, \dots,x_n^{(c)}\in \{0,1\}^*$ given input $1^n$, where $|x_n^{(i)}| = \ell^{(i)}(n)$. 
Then, for every $i\in [c]$, define $B^{(i)}$ to be the algorithm that prints $x^{(i)}_n$ on input $1^{\ell^{(i)}(n)}$. Note that algorithm $B^{(i)}$ is well-defined because $\ell^{(i)}(n)$ is strictly increasing (and hence injective), and it runs in $\poly(n)\le \poly(\ell^{(i)}(n))$ time.
By \cref{defn:constsizelistref}, observe that at least one of $B^{(1)},B^{(2)},\dots, B^{(c)}$ prints valid counterexamples for infinitely many $n$. Hence, there exists at least one $i\in [c]$ such that $B^{(i)}$ is a $\PTIME$-refuter (\cref{defn:refuter}).
	
Similarly, if a constant-size $\BPP$ (or $\ZPP$) list-refuter is pseudo-deterministic, then the same argument also applies, and we can obtain a $\BPP$ (or $\ZPP$) refuter.
\end{proof}

\subsection{Refuters for PSPACE, EXP, and NEXP}
\label{sec:refuter-nexp}

We first consider the case when $\caD$ is a complexity class from $\{\PSPACE, \EXP, \NEXP\}$. Our proof below generalizes the refuter construction of~\cite{DolevFG13} which only discussed the case of $\calD = \NEXP$.

Let $(\exists \poly(n))\caD$ denote the complexity class that contains languages $L$ satisfying the following property: there exists a polynomial $p(n)$ and a language $L' \in \caD$ such that for all strings $x$, $L(x)=1$ if and only if there exists $y\in \{0,1\}^{p(|x|)}$ such that $L'(x,y)=1$. Similarly, we define the complexity class $(\forall \poly(n))\caD$.

\begin{theorem}
\label{thm:list-nexp}
	Let $\caC\in \{\PTIME,\BPP,\ZPP\}$, and $\caD$ be a complexity class such that $\caC\subseteq \caD$.
	Suppose~$\caD$ satisfies $(\exists \poly(n)) \caD \subseteq \caD$ and  $(\forall \poly(n)) \caD \subseteq \caD$.

	If $\caD \nsubseteq \caC$, then for every paddable $\caD$-complete language $L$, there is a $\caC$-constructive separation of $L\notin \caC$. 
\end{theorem}

\begin{proof}
	We first consider the case of $\caC=\PTIME$. Let $A$ be a polynomial-time algorithm. 
 Let $n$ be an input length such that $A$ does not correctly solve $L$ on all $n$-bit inputs; since $\caD\nsubseteq \PTIME$, we know there are infinitely many such input lengths. %
	For $b\in \{0,1\}$, define the language 
	$$G_A^{(b)}:= \{  (1^n, x)  :  \text{there exists $y\in \{0,1\}^n$ with prefix $x$, such that $L(y)=b, A(y)=1-b$}\}.$$
	Observe that $G_A^{(1)}\in (\exists \poly(n)) \caD\subseteq \caD$, $G_A^{(0)}\in (\exists \poly(n)) \co \caD \subseteq \co \caD$. Define \[G_A:= G_A^{(0)}\cup G_A^{(1)} =\{  (1^n, x)  :  \text{there exists $y\in \{0,1\}^n$ with prefix $x$, such that $L(y)\neq A(y)$}\}.\]
	
	Since $L$ is $\caD$-complete, there is a polynomial-time procedure $R^L$ that can decide  $G_A$ by making two queries to an oracle for $L$. Since $L$ is paddable, we may assume the queries to the $L$-oracle always have length exactly $\ell(n)$, for some strictly increasing polynomial $\ell\colon \N\to \N$.
	If we let $R$ query the algorithm $A$ instead of the $L$-oracle, then on any $(1^n,x)$, either $R^A$ solves $G_A(1^n,x)$ correctly, or $A$ gives the incorrect answer on at least one of the queries. 
	
	Our list-refuter performs a search-to-decision reduction which repeatedly calls $R^A(1^n,x)$ and extends the prefix $x$ one bit at a time. It either eventually finds a string $y\in \{0,1\}^n$ such that $L(y)\neq A(y)$, or detects the inconsistency of $A$'s answers.
	The pseudocode of this list-refuter is presented in Algorithm~\ref{algo:list}.
	\begin{figure}[ht]
	\renewcommand{\figurename}{Algorithm} 
	 \caption{The list-refuter against $A$}
	\label{algo:list}
	\begin{itemize}[noitemsep,parsep=0pt,partopsep=0pt,topsep=0pt]
		\item  Initialize $x$ as an empty string
		\item For $i \gets 1,2,\dots,n$:
		\begin{itemize}[noitemsep,parsep=0pt,topsep=0pt,partopsep=0pt]
			\item If $R^A(1^n, x\circ 1 ) = 1$:
			\begin{itemize}[noitemsep,parsep=0pt,topsep=0pt,partopsep=0pt]
				\item $x\gets x\circ 1$
			\end{itemize}
			\item Else if $R^A(1^n,x\circ 0)=1$:
			\begin{itemize}[noitemsep,parsep=0pt,topsep=0pt,partopsep=0pt]
				\item $x\gets x\circ 0$
			\end{itemize}
			\item Else:
			\begin{itemize}[noitemsep,parsep=0pt,topsep=0pt,partopsep=0pt]
				\item  Return all the queries sent to $A$ by  $R^A(1^n, x), R^A(1^n, x\circ 0),$ and  $R^A(1^n, x\circ 1)$
			\end{itemize}
		\end{itemize}
		\item Return $x$ and all the queries sent to $A$ by $R^A(1^n,x)$
	\end{itemize}
\end{figure}
	
	To prove the correctness of this list-refuter, we suppose for contradiction that $A$ could correctly solve every string in the list.
	Consider three cases according to the final length of $x$ when the refuter terminates:
	\begin{enumerate}[label=(\arabic*)]
		\item $|x|=0$. Then $(1^n, 1)$ and $(1^n, 0)$ are not in $G_A$, which is impossible, since $A$ cannot solve $L$ correctly on every $n$-bit input. %
		\item $1\le |x|<n$. Then $(1^n, x)\in G_A$, but $(1^n, x\circ 1)$ and $(1^n, x\circ 0)$ are not in $G_A$. This is also impossible.
		\item $|x|=n$. Then $(1^n, x)\in G_A$, meaning that $L(x)\neq A(x)$.
		But $x$ is also in the list and $A$ should solve $x$ correctly, a contradiction.
	\end{enumerate}
	Hence, $A$ answers incorrectly on at least one string in the list returned by Algorithm~\ref{algo:list}.
	
	The list contains at most six strings, each of which has length $n$ or $\ell(n)$.
	By  \cref{rem:list}, this constant-size list-refuter can be converted into a refuter.

Now we consider the case of $\caC = \BPP$.	Since $A\in \BPP$, by standard amplification\footnote{We remark that \cite{GutfreundST07} also studied the case where $A$ does not have bounded probability gap, which we do not consider here.}, there is  another $\BPP$ algorithm $A'$ which decides the same language as $A$ and has success probability $1-2^{-2n}$. Then, for a uniformly chosen random seed $r$, with $1-2^{-n}$ probability, $A'(\cdot,r)$ decides the same language as $A$ on input length $n$. From this point, we may apply the same proof of the $\caC=\PTIME$ case to $A'(\cdot,r)$.  Hence we have a $\BPP$-refuter against $A$.  If we further assume $A\in \ZPP$, then the refuter also has zero error. Note that our randomized refuters are pseudo-deterministic.
\end{proof}
\begin{corollary}
	Let $(\caC, \caD)$ be a pair of complexity classes from
	\[
	\{ \PTIME, \ZPP,\BPP \} \times \{\PSPACE, \EXP,\NEXP,\EXP^\NP\}.
	\]
	Assuming $\caD \nsubseteq \caC$, for every paddable $\caD$-complete language $L$, there is a $\caC$-constructive separation of $L \notin \caC$.
\end{corollary}
\begin{proof}
Note that all pairs $(\caC,\caD)$ satisfy the requirements in \cref{thm:list-nexp} (where the inclusion $(\forall \poly(n)) \NEXP \subseteq \NEXP$ follows from concatenating the witnesses for every possibility in the universal quantifier).
\end{proof}

\subsection{Refuters for NP and the Polynomial Hierarchy} 

Now we move to the case that $ \calD = \Sigma_k \PTIME$ for an integer $k$.

\label{sec:refuter-np}
\begin{theorem}[Adaptation of \cite{GutfreundST07}]
\label{thm3.5}
		Let $\caC\in \{\PTIME,\BPP,\ZPP\}$.
	Suppose  $\NP \subseteq \caD$, and there is a  $\caD$-complete language $M$ which is downward self-reducible.
	
	If $\caD \nsubseteq \caC$, then for every paddable $\caD$-complete language $L$, there is a $\caC$-constructive separation of $L\notin \caC$. 
\end{theorem}

\begin{proof}[Proof Sketch]
 Let $A$ be any algorithm in $\caC$. We will construct a refuter for $L$ against $A$.  Here we only prove the case of $\caC= \PTIME$. (For $\caC \in \{\BPP,\ZPP\}$, we use the same proof as the $\caC=\PTIME$ case, and apply the amplification argument described at the end of the proof of \cref{thm:list-nexp}.)

	Since $M$ is downward self-reducible, there is a polynomial-time procedure $D$ such that for every $x\in \{0,1\}^m$, $M(x) = D^{M_{\le m-1}}(x)$.
 
	Since $L$ is $\caD$-complete and $M\in \caD$, there is a $\poly(n)$-time reduction $p_n\colon \{0,1\}^{n} \to \{0,1\}^{q(n)}$ such that $M(x) = L(p_n(x))$, where $q(n)\colon \N \to \N$ is some strictly increasing polynomial.
 For convenience of later proof, we extend the domain of $p_n$ to $p_n\colon \{0,1\}^{\le n} \to \{0,1\}^{q(n)}$ so that $M(x) = L(p_n(x))$ holds for $|x|<n$ as well. This can be done by first mapping $x$ to $p_{|x|}(x)$, and then use the paddability of $L$  to pad $p_{|x|}(x)$ to a string of length $q(n)$.
	
	For large enough $n$, there must exist an $x \in \{0,1\}^{\le n}$ such that $A(p_n(x)) \neq M(x)$, since otherwise we would have a $\caC$ algorithm that decides $M$, contradicting $\caD \nsubseteq \caC$.
	Then we argue that there must be a string $x$ of length $m\le n$, such that
	\begin{equation}
	\label{eq:condition}
	A(p_n(x)) \neq D^{O_{m-1}}(x), \text{ where } O_{m-1}:= \{x \in \{0,1\}^{\le m-1}: A(p_n(x)) = 1\},
	\end{equation}
	since otherwise the definition of $D$ (by downward self-reducibility of $M$) together with an induction on $m$ would imply $A(p_n(x)) = M(x)$ for all $x\in \{0,1\}^{\le n}$, contradicting $A(p_n(x))\neq M(x)$. 
 Let $x^*$ be the shortest string $x$ satisfying condition (\ref{eq:condition}).
 Then the minimality of $|x^*|$ implies $D^{O_{|x^*|-1}}(x^*) = M(x^*)$, and hence  $A(p_n(x^*)) \neq M(x^*)$ by (\ref{eq:condition}). 
 Then from $M(x^*)=L(p_n(x^*))$ we know $A(p_n(x^*)) \neq L(p_n(x^*))$, so $p_n(x^*)$ is a counterexample showing $A$ does not solve $L$.
	
Observe that condition (\ref{eq:condition}) can be checked in polynomial time, so such $x^*$ can be found if we had an $\NP$ machine. Since $A$ claims to decide an $\NP$-hard language, we can try to find $x^*$ by a search-to-decision reduction using $A$, similarly to what we did in the proof of \cref{thm:list-nexp}.
More specifically, our list-refuter does the following:
\begin{itemize}
    \item First consider the oracle algorithm
$R_0^A(1^{m'},1^n)$ which claims to solve the following $\NP$ question by making one query to $A$: ``does there exist a string $x$ of length $|x|=m\le m'$ such that condition (\ref{eq:condition}) is satisfied by $n,x$, and $m$?'' 
The answer to $R_0^A(1^{n},1^n)$ is supposed to be YES; if it returns NO, then we know the query made to $A$ by $R_0^A(1^{n},1^n)$ is a counterexample, and we are done. Otherwise, we find the smallest length $m\le n$ such that $R_0^A(1^{m},1^n)$ returns YES but $R_0^A(1^{m-1},1^n)$ returns NO.  Then $m$ is supposed to be the length of the shortest $x^*$.
\item Then, consider the oracle algorithm $R^A(y,1^m,1^n)$ which claims to solve the following $\NP$ question: ``does there exist $x\in \{0,1\}^m$ whose prefix is $y$ such that condition (\ref{eq:condition}) is satisfied by $n,x$, and $m$?''
We gradually extend the prefix $y$ in the same way as in \cref{thm:list-nexp}, until either we find inconsistency between the answers for $y, y\circ 0, y\circ 1$, or we eventually extend it to full length $|y|=m$ and obtain $x^*=y$. 
In the former case, we add the queries made to $A$ by $R^A(y,1^m,1^n),R^A(y\circ 0,1^m,1^n),R^A(y\circ 1,1^m,1^n)$ into our list of counterexamples.
In the latter case, if $R^A(x^*,1^m,1^n)$ returns YES but condition (\ref{eq:condition}) is not satisfied, then we know $R^A(x^*,1^m,1^n)$ made a mistake and we also get a counterexample.

The remaining case is where we obtained an $x^*$ of length $m$ that indeed satisfies condition (\ref{eq:condition}). In this case,
we add $p_n(x^*)$ to the list of counter examples. We also need to add the query made to $A$ by $R_0^A(1^{m-1},1^n)$ to the list of counterexamples. By our discussion earlier, if $x^*$ is the shortest string satisfying condition (\ref{eq:condition}), then $p_n(x^*)$ is a counterexample. If $x^*$ is not the shortest, then the answer NO returned by $R_0^A(1^{m-1},1^n)$ is a mistake, and we also get a counterexample.
\end{itemize} 
Hence we have designed a constant size list-refuter for $L$ against $A$, and the rest of the proof follows in the same way as \cref{thm:list-nexp}.
\end{proof}
We can compare this proof with  the earlier proof of \cref{thm:list-nexp}. In \cref{thm:list-nexp}, we used an $(\exists \poly(n))\caD$ machine to find counterexamples for a $\caD$ problem, so we needed the assumption that $\caD$ is closed under $\exists$ (and $\forall$).  Here in \cref{thm3.5}, we side-step this $\exists$-closure assumption by using the downward self-reducibility of $\caD$ instead. In this way, we get a polynomial-time checkable condition (\ref{eq:condition}), which allows us to find a counterexample using only an $\NP$ machine.

The following corollary follows immediately from Theorem~\ref{thm3.5} and the fact that $\Sigma_k \PTIME$ has a downward self-reducible complete language $\Sigma_k \SAT$.

\begin{corollary}
	Let $(\caC, \caD)$ be a pair of complexity classes from the following list 
	\[
	\{ \PTIME, \ZPP,\BPP \} \times \{\Sigma_k \PTIME\}_{k \ge 1}.
	\]
	If $\caD \nsubseteq \caC$, then for every paddable $\caD$-complete language $L$, there is a $\caC$-constructive separation of $L\notin \caC$. 
\end{corollary}

\subsection{Refuters for PP and Parity-P} 
\label{sec:refuter-pp}

Finally we prove Theorem~\ref{theo:refuter-uniform-sep} for the case $\calD \in \{ \PP,\ParP \}$.

\begin{theorem}
\label{thm:pp}
Let $\caC\in \{\PTIME,\BPP,\ZPP\}$.
If $\PP \nsubseteq \caC$,  then for every paddable $\PP$-complete language~$L$, there is a $\caC$-constructive separation of $L \notin \caC$.
\end{theorem}
\begin{proof}
 Let $A$ be any $\caC$-algorithm. We will construct a refuter for $L$ against $A$. 
Here we only prove the case of $\caC= \PTIME$. (For $\caC \in \{\BPP,\ZPP\}$, we use the same proof as the $\caC=\PTIME$ case, and apply the amplification argument described at the end of the proof of \cref{thm:list-nexp}.)

We first review the well-known polynomial-time algorithm $D^{\PP}$ that solves $\ssat$ with the help of a $\PP$ oracle.
Given a 3-CNF formula $\phi$ with $n$ variables, let $c_{n}c_{n-1}\cdots c_0$ denote the number of satisfiable assignments of $\phi$ in binary, i.e.,
 \[\ssat(\phi) = \sum_{0\le i\le n}c_i \cdot 2^i,\] where $c_i \in \{0,1\}$ for $i=0,\ldots,n$.
The algorithm computes the values of $c_i$ in decreasing order of $i$: 
after $c_n,c_{n-1},\dots,c_{i+1}$ are determined, $c_i$ is the truth value of the statement \[\ssat(\phi) \ge 2^i+\sum_{i+1\le j\le n}c_j\cdot 2^j,\] which can be determined by the $\PP$ oracle. Hence $D^{\PP}$ can compute $\ssat(\phi)$ using $n+1$ queries to a $\PP$ oracle. 
Observe that this query algorithm $D^{\PP}$ must have asked the queries \[\ssat(\phi) \ge  \sum_{i\le j\le n} c_j\cdot 2^j\] for all $0\le i<n$,  and the oracle answers $1$ to these queries. (For example, if $n=4$ and $c_4c_3c_2c_1c_0=01011$, then $D^{\PP}$ asked queries ``$\ssat(\phi)\ge x$'' for $x\in \{10000,01000,01100,$ $01010,01011\}$.) Similarly, observe that 
$D^{\PP}$ must have asked the query
 \[\ssat(\phi) \ge  1+\sum_{0\le j\le n} c_j\cdot 2^j,\] to which the oracle answered $0$.

Since $A$ claims to decide a $\PP$-complete language, we replace the $\PP$ oracle by $A$ and try to use $D^{A}$ to solve  $\ssat$ on $n$ variables.
By padding, we assume the input strings received by $A$ have length exactly $\ell(n)$, for some strictly increasing polynomial $\ell\colon \N \to \N$.
The polynomial-time algorithm $D^A$ cannot correctly solve $\ssat$ on all possible $\phi$, since otherwise it would contradict the assumption that $\PP\nsubseteq \PTIME$.
Hence there exists a formula $\phi$ such that $D^A(\phi) \neq D^A(\phi_0)+D^A(\phi_1)$, where 
 $\phi_b$ denotes the formula obtained by setting the first variable in $\phi$ to $b$.  Since $\NP \subseteq \PP$, we can try to find such a $\phi$ by a search-to-decision reduction using $A$, analogously to the proof of \cref{thm3.5} and \cref{thm:list-nexp}.
We either find such a $\phi$, or detect inconsistency during the search and find a constant-size list that contains a counterexample.

Now suppose we have found such a $\phi$ with $m$ variables satisfying $D^A(\phi) \neq D^A(\phi_0) + D^A(\phi_1)$. Then we know that $A$ answered incorrectly on one of the $(m+1)+m+m=3m+1$ queries 
asked by $D$. In the following we show how to reduce the size of this list to $O(1)$.

Let $a_{m}\cdots a_0$, $b_m\cdots b_0$, $c_{m}\cdots c_0$ be the binary representation of $D^A(\phi_0),D^A(\phi_1)$, and $D^A(\phi)$, respectively (where $a_m=b_m=0$).
Since $D^A(\phi) \neq D^A(\phi_0) + D^A(\phi_1)$, we know
\[\sum_{0\le j\le m}(c_j-a_j-b_j)\cdot 2^j\ne 0.\]  
We assume $D^A(\phi)\le 2^m$ (and similarly, $D^A(\phi_0),D^A(\phi_1) \le 2^{m-1}$); otherwise $A$ must have answered $1$ to the query ``$\ssat(\phi)\ge S$'' for some $S\ge 2^m+1$, which is an obvious counterexample.  Now consider two cases:
\begin{enumerate}[label=(\arabic*)]
    \item  $\sum_{0\le j\le m}(c_j-a_j-b_j)\cdot 2^j \le -1$. 
    We know that $A$ answered 0 to the query ``$\ssat(\phi)\ge 1+ \sum_{0\le j\le m}c_j\cdot 2^j$'', and answered 1 to the queries ``$\ssat(\phi_0)\ge \sum_{0\le j\le m}a_j\cdot 2^j$'' and ``$\ssat(\phi_1)\ge \sum_{0\le j\le m}b_j\cdot 2^j$''. Assuming that all three answers are correct, we have
    \begin{align*}
        0 &= \ssat(\phi) - \ssat(\phi_0)-\ssat(\phi_1)\\
        & < (1+ \sum_{0\le j\le m}c_j\cdot 2^j) -(\sum_{0\le j\le m}a_j\cdot 2^j) - (\sum_{0\le j\le m}b_j\cdot 2^j )\\
        & = 1 + \sum_{0\le j\le m}(c_j-a_j-b_j)\cdot 2^j\\
        & \le 0,
    \end{align*}
    a contradiction.
    
    \item  
    $\sum_{0\le j\le m}(c_j-a_j-b_j)\cdot 2^j >0$.
        We know that $A$ answered 1 to the query ``$\ssat(\phi)\ge \sum_{0\le j\le m}c_j\cdot 2^j$'', and answered 0 to the queries ``$\ssat(\phi_0)\ge 1+\sum_{0\le j\le m}a_j\cdot 2^j$'' and ``$\ssat(\phi_1)\ge 1+ \sum_{0\le j\le m}b_j\cdot 2^j$''. Assuming that all three answers are correct, we have
    \begin{align*}
        0 &= \ssat(\phi) - \ssat(\phi_0)-\ssat(\phi_1)\\
        & \ge  (\sum_{0\le j\le m}c_j\cdot 2^j) -(\sum_{0\le j\le m}a_j\cdot 2^j) - (\sum_{0\le j\le m}b_j\cdot 2^j )\\
        & = \sum_{k\le j\le m}(c_j-a_j-b_j)\cdot 2^j\\
        & > 0,
    \end{align*}
    a contradiction.
\end{enumerate}
In either of the two cases, we obtain a list of three strings that contains at least one counterexample. This finishes our construction of the constant-size list-refuter, which can be converted into a refuter by applying \cref{rem:list}.
\end{proof}

\begin{theorem}
Let $\caC\in \{\PTIME,\BPP,\ZPP\}$.
If $\ParP \nsubseteq \caC$,  then for every paddable $\ParP$-complete language~$L$, there is a $\BPP$-constructive separation of $L \notin \caC$.
\end{theorem}
\begin{proof}[Proof Sketch]
 Let $A$ be any algorithm in $\caC$. We will construct a refuter for $L$ against $A$. 
Here we only prove the case of $\caC= \PTIME$. (For $\caC \in \{\BPP,\ZPP\}$, we use the same proof as the $\caC=\PTIME$ case, and apply the amplification argument described at the end of the proof of \cref{thm:list-nexp}.)

The proof is similar to that of \cref{thm:pp}. Let $R$ be a reduction from $\psat$ to $L$. Then there must exist a 3-CNF formula $\phi$ such that
$A(R(\phi)) \neq A(R(\phi_0)) \oplus A(R(\phi_1))$, where $\phi_b$ denotes the formula obtained by setting the first variable in $\phi$ to $b$. If we can find such $\phi$, then we immediately obtain three strings which contain a counterexample for $A$.

 Our requirement for $\phi$ can be encoded as a $\SAT$ instance $\pi$. By the Valiant-Vazirani theorem \cite{vv} and the fact that $\PTIME^{\oplus\PTIME}=\oplus\PTIME$~\cite{PapadimitriouZ83}, there is a polynomial-time reduction $f$ with random seed $r$ such that if $x\notin \SAT$, then $\Pr_r[f(x,r) \notin \psat] = 1$, and if $x\in \SAT$, then $\Pr_r[f(x,r)\in  \psat] \ge 2/3$.\footnote{In more detail, Valiant-Vazirani says that there is a randomized Turing reduction from $\SAT$ to $\oplus \SAT$ such that a given formula $x$ is reduced to a sequence of formulas $x_1,\ldots,x_{O(n)}$ which are called on $\oplus \SAT$. We take the entire Turing reduction from $\SAT$ to $\oplus\SAT$, with success probability increased to at least $2/3$, and apply the fact that $\PTIME^{\ParP} = \ParP$, to obtain a single $\oplus \SAT$ instance.}
We pick a random seed $r$, and consider two cases:
\begin{itemize}
    \item 
If $A(R(f(\pi,r))) = 1$, then we can use the downward self-reducibility of $\psat$ to perform a search-to-decision reduction using $A$ (similar to the proof of \cref{thm:list-nexp}). Either we find a satisfying assignment for $f(\pi,r)$, or we detect that $A$'s answers are inconsistent. 
In the first case, note that the reduction $f$ of \cite{vv} is simple enough so that we can efficiently convert any satisfying assignment for $f(\pi,r)$ to a satisfying assignment for $\pi$, which can then be converted to a formula $\phi$ that satisfies the desired property $A(R(\phi)) \neq A(R(\phi_0)) \oplus A(R(\phi_1))$.
\item 
If $A(R(f(\pi,r))) = 0$, then our refuter simply outputs $R(f(\pi,r))$. Observe that this string is indeed a counterexample for $A$ if $f(\pi,r) \in \psat$, which happens with probability at least $2/3$. \qedhere
\end{itemize}
\end{proof}

\paragraph*{Remark: These refuters are non-black-box.} Observe that all refuter constructions in this section do require access to the \emph{code} of the algorithm $A$ being refuted. (That is, our refuter constructions are not ``black-box'' in terms of the algorithm $A$.) Atserias~\cite{Atserias06} constructed a black-box refuter for the separation $\NP \not\subset \BPP$ (more strictly speaking, Atserias’ refuter is only ``grey-box'' in that it needs to know the running time of the $\BPP$ algorithm it fools). It may be possible to improve our refuter constructions to be black-box (or ``grey-box'') as well. 
However, it seems  challenging to use the techniques of~\cite{Atserias06} for this, because he crucially relies on the $\ZPP^{\NP}$ learning algorithm for polynomial-size circuits~\cite{BshoutyCGKT96}. It is unclear how one might prove $\PTIME$-constructive separations using such an algorithm.

\section{Hard Languages With No Constructive Separations}\label{sec:hard-language-no-constructive-sep}

In this section we show there are hard languages without constructive separation from any complexity class. We first observe there are no constructive separations for $\RKt$ unconditionally.

\propnorefuter*

% \begin{theorem}[\cref{prop:norefuter-RKt} restated]
% 	For any $t(n) \ge n^{\omega(1)}$, there is no $\PTIME$-refuter for $\RKt$ against the constant zero function.
% \end{theorem}
\begin{proof}
A $\PTIME$ refuter for $\RKt$ against the constant zero function needs to output in $\poly(n)$ time an $n$-bit string $y_n$ with $\Ksupt$ complexity at least $n-1$, for infinitely many integers $n$. But by the definition of $\Ksupt$ complexity, all these $y_n$ can be computed in $\poly(n)$ time by a uniform algorithm given the input $n$ of $\log n$ bits, hence $\Ksupt(y_n) = O(\log n)$ for all $n$, a contradiction. 
\end{proof}

Next we show that, under plausible conjectures, there are languages in $\NP \setminus \PTIME$ with no constructive separations from any complexity class.

\renewcommand{\footnotenonpb}{\footnote{Recall we have defined $\RE$ to be one-sided randomized time $2^{O(n)}$.}}
\theononp*
% \begin{theorem}[\cref{thm:no-np-refuter} restated]
% 	The following hold:
% 	\begin{itemize}
% 		\item 
% If $\NE \neq \E$, then there is a language in $\NP \backslash \PTIME$ that does not have $\PTIME$ refuters against the constant one function.
% 		\item 
% If $\NE \neq \RE$, then there is a language in $\NP \backslash \PTIME$ that does not have $\BPP$ refuters against the constant one function.\footnote{Recall we have defined $\RE$ to be one-sided randomized time $2^{O(n)}$.} 
% 	\end{itemize}
% \end{theorem}
\begin{proof}
Assume $\NE\neq \E$, and let $L' \in \NE \backslash \E$.
Suppose for some constant $c\ge 1$ there is a $2^{O(n)}$ time reduction $R\colon \{0,1\}^n \to \{0,1\}^{2^{cn}}$ such that $x\in L' \Leftrightarrow R(x) \in \SAT$.

We define a language $L$ as follows: 

\begin{itemize}
	\item For $m\in \N$, $L$ is given the concatenated string \[(t, w_0,w_1,\dots,w_{2^m-1}, s) \in \{0,1\}^{2^m} \times \big (\{0,1\}^{2^{cm}})^{2^m} \times \{0,1\}^{2^{c(m+1)}}\] as input. 
	
	Here, $m$ is intended as the input length to the language $L'$, $t$ is interpreted as a potential truth table of $L'$ on all $m$-bit inputs which needs to be verified, $w_0,\dotsc,w_{2^m-1}$ are interpreted as potential witnesses for every $m$-bit inputs to $L'$ to help the verification, and $s$ is intended as an input to $\SAT$.
	
	\item $L(t, w_0,w_1,\dots,w_{2^m-1}, s) = 1$ if and only if all of the following conditions hold:
	\begin{enumerate}[label=(\arabic*)]
		\item For every $i\in \{0,1\}^m$ with $t_i=1$, we have that  $w_i \in \{0,1\}^{2^{cm}}$ is a correct witness of $R(i)\in \SAT$ (in particular, $i\in L'$).
		\item For every $i\in \{0,1\}^m$ with $t_i=0$, we have $i\notin L'$.
		\item $s \notin \SAT$.
	\end{enumerate}
\end{itemize}

That is, $L$ accepts the input $L(t, w_0,w_1,\dots,w_{2^m-1}, s)$ if $t$ is the correct truth table of $L'$ on all $m$-bit inputs and all the $w_i$ are correct witnesses for the corresponding inputs to $L'$, and $s \notin \SAT$.  %

The conditions (1) and (2) above mean that every input accepted by $L$ \emph{must reveal the truth table of the language $L'$}, which helps us to design an $\ETIME$ (or $\RE$) algorithm for $L'$ given a $\PTIME$ (or $\BPP$) refuter for $L$.
Condition (3) allows us to argue that if $\PTIME \neq \NP$, then $L' \notin \PTIME$.

The concatenated string has length $2^{\Theta(m)}$. We can verify condition (1) in $2^{O(m)}$ time, and verify conditions (2) and (3) in $\co \NTIME[2^{O(m)}]$, so  $L\in \coNP$. 

\begin{claim}
\label{claim1}
If $L\in \PTIME$, then $\SAT \in \PTIME$.
\end{claim}
From Claim \ref{claim1} we conclude $L\notin \PTIME$, since otherwise it would imply $\PTIME=\NP$ and consequently $\E = \NE$, contradicting our assumption.
Hence, $\overline{L} = \{0,1\}^* \backslash L$ is a language in $\NP \backslash \PTIME$.

We will show that $\overline{L}$ does not have $\PTIME$ refuters against the constant one function. 
If there is such a refuter, then it must output in $2^{O(m)}$ time a string $(t, w_0,w_1,\dots,w_{2^m-1}, s)\in L$. By the conditions (1) and (2) in the definition of $L$, we have $t_i = L'(i)$ for all $i\in \{0,1\}^m$. Hence, we can use this refuter to decide $L'$ on $m$-bit inputs in $2^{O(m)}$ time, contradicting $L' \notin \E$.

To prove the second statement of the theorem, we further assume $\NE \neq \RE$ and $L'\in \NE \backslash \RE$.
Suppose $\overline{L}$ has a $\BPP$ refuter against the constant one function, which prints a string $(t, w_0,w_1,\dots,w_{2^m-1}, s)$. With at least $2/3$ probability, the string is in $L$.
On a given input $i\in \{0,1\}^m$, if $t_i=1$ and $w_i$ is a correct witness of $i\in L'$, then we return $L'(i)=1$; otherwise, we return $L'(i)=0$. This yields a one-sided error randomized algorithm that decides  $L'$  on $m$-bit inputs in $2^{O(m)}$ time, contradicting $L' \notin \RE$.

It remains to prove Claim \ref{claim1}.
\begin{subproof}[Proof of Claim \ref{claim1}]

Recall that $L' \in \NE$ and  $R\colon \{0,1\}^n \to \{0,1\}^{2^{cn}}$ is a  $2^{O(n)}$ time reduction such that $y\in L' \Leftrightarrow R(y) \in \SAT$.
Assume $L\in \TIME[n^{d}]$.
The recursive algorithm \texttt{Solve-SAT} (described in Algorithm~\ref{algo:solvesat}) receives $m\in \N$ and $x\in \{0,1\}^{2^{cm}}$ as input, and outputs a pair $(\SAT(x), w)$, 
where $w\in \{0,1\}^{2^{cm}}$ is a correct witness if $\SAT(x)=1$.

	\begin{figure}[ht]
	\renewcommand{\figurename}{Algorithm} 
	 \caption{Solve-SAT}
	 \label{algo:solvesat}
\texttt{Solve-SAT}($m,x$) :
\begin{itemize}[noitemsep,parsep=0pt,topsep=0pt,partopsep=0pt]
    \item If $m\le O(1)$, then return the correct $(\SAT(x),w)$ in constant time
    \item For $y \in \{0,1\}^{m-1}$:
    \begin{itemize}[noitemsep,parsep=0pt,topsep=0pt,partopsep=0pt]
        \item Let $(t_y,w_y) := $ \texttt{Solve-SAT}$(m-1,R(y))$
    \end{itemize}
    \item Let $\texttt{answer}:=  \overline{L}(t,w_0,w_1,\dots, w_{2^{m-1}-1}, x)$ 
    \item If $\texttt{answer} = 1$, then find a correct witness $w$ of $x\in \SAT$ by a search-to-decision reduction which repeatedly calls $\overline{L}(t,w_0,\dots,w_{2^{m-1}-1},\cdot)$
    \item Return $(\texttt{answer},w)$
\end{itemize}
\end{figure}
The correctness of \texttt{Solve-SAT} easily follows from the definition of $L$ and an induction on $m$.
The overall idea of this algorithm is to use  $\overline{L}(t,w_0,\dots,w_{2^{m-1}-1},\cdot)$ as a $\SAT$ solver after we obtain the correct $t, w_0,\dots,w_{2^{m-1}-1}$, which themselves can be found by solving smaller $\SAT$ questions.

To improve the running time of the algorithm, we implement \texttt{Solve-SAT} with memoization. That is, if $(t_y,w_y)$ at the $m$-th level of the recursion is already computed, then later it can be directly accessed without recursively calling \texttt{Solve-SAT} again.
Then, the total time of \texttt{Solve-SAT} is at most $\sum_{m'\le m} 2^{m'-1}\cdot 2^{cm'}\cdot (2^{O(m')})^d \le 2^{O(m)}$. Hence, we can solve $\SAT(x)$ in $\poly(|x|)$ time.
\end{subproof}
This completes the proof of the overall theorem.
\end{proof}

It is interesting to contrast \cref{thm:no-np-refuter} with \cref{thm3.5}, which says $\PTIME\neq \NP$ implies that  every paddable $\NP$-complete language has a $\PTIME$-constructive separation of $L\notin \PTIME$. This means the language $\overline{L}\in \NP\setminus \PTIME$ in \cref{thm:no-np-refuter} is not $\NP$-complete.

\section{Conclusion}\label{sec:future-work}
Many interesting questions remain for future work. While we have given many examples of complexity separations that can automatically be made constructive, it is unclear how to extend our results to separations with complexity classes within $\PTIME$. For example, let $L$ be a $\PTIME$-complete language. If $L$ is not in uniform $\NC^1$, does a $\PTIME$-constructive separation of $L$ from uniform $\NC^1$ follow? How about separations of $\PTIME$ from ${\sf LOGSPACE}$? Would establishing constructive separations in these lower complexity classes have any interesting consequences?

Note that there is no $\PTIME$-constructive separation of $\MCSP[s] \notin \PTIME$ for super-polynomially large $s$, unless $\EXP$ requires super-polynomial size Boolean circuits. (A polynomial-time refuter for the trivial algorithm that always accepts, must print a hard function!) But do any interesting consequences follow from a constructive separation of \emph{search versions} of $\MCSP$ from $\PTIME$? The same proof strategy (of applying the conjectured refuter for the trivial algorithm that always accepts) %
does not make sense in this case, as the only hard instances for search problems are YES instances.

It would also be interesting to consider constructive separations against \emph{non-uniform} algorithms.
We say a $\PTIME$ list-refuter $R$ for a language $L$ against circuit class $\caC$ is a deterministic polynomial time algorithm that, given the description of a circuit $C_n$ on input length $n$ where $\{C_n\}_{n\in \mathbb{N}}\in \caC$, finds a list of $x_i\in \{0,1\}^n$ such that $L(x_i)\neq C(x_i)$ for some $i$, for infinitely many input lengths $n$. 
We also say that $R$ gives a $\PTIME$-constructive separation $L \notin \caC$.
Furthermore, we say it is an oblivious list-refuter, if it does not need access to the description of the circuit $C_n$ (this was called explicit obstructions in \cite{cjw20}).
It would be interesting to examine which proof methods for circuit lower bounds can be made constructive. We list a few examples which should be particularly interesting:
\begin{itemize}
    \item[(1)] the $\widetilde{\Omega}(n^3)$ size lower bound against DeMorgan formulas for Andreev's function~\cite{Hastad98,Tal14},
    \item[(2)] the $\widetilde{\Omega}(n^2)$ size lower bound against formulas for Element-Distinctness~\cite{Neciporuk1691966},
    \item[(3)] $\AC^0[p]$ size-depth lower bounds via the approximation method~\cite{razborov1987lower,Smolensky87}. 
\end{itemize}
Chen, Jin, and Williams~\cite{cjw20} showed that constructing corresponding explicit obstructions for (1) and (2) above would imply $\EXP \not\subset \NC^1$, but it is unclear whether one can get a $\PTIME$-constructive separation without implying a major breakthrough lower bound. 

We remark that as shown in~\cite{cjw20}, most lower bounds proved by random restrictions \emph{can} be made constructive, by constructing an appropriate pseudorandom restriction generator.~\cite{cjw20} explicitly constructed an oblivious list-refuter for parity against subquadratic-size formulas, and we remark that a similar oblivious list-refuter for parity against polynomial-size $\AC^0$ circuits follows from the pseudorandom restriction generator for $\AC^0$ of~\cite{gw14}.

Atserias \cite[Theorem 3]{Atserias06} showed that 
$\NP \not\subset \PTIME/\poly$ implies a $\BPP$-constructive separation $\NP \not\subset \PTIME/\poly$ (note that Atserias' refuters only need to know the size of the circuits being refuted).
 An interesting open problem following the work of Atserias is whether separations of the form $\mathcal{C} \not \subset \PTIME/\poly$ can be made constructive for classes $\mathcal{C}$ higher than $\NP$ (for example, $\NEXP$).

%\bibliographystyle{alpha}

%\bibliography{main}
 
 	\printbibliography

@inproceedings{ChenT21b,
  author       = {Lijie Chen and
                  Roei Tell},
  title        = {Hardness vs Randomness, Revised: Uniform, Non-Black-Box, and Instance-Wise},
  booktitle    = {62nd {IEEE} Annual Symposium on Foundations of Computer Science, {FOCS}
                  2021, Denver, CO, USA, February 7-10, 2022},
  pages        = {125--136},
  publisher    = {{IEEE}},
doi          = {10.1109/FOCS52979.2021.00021},
  year         = {2021}
}

@inproceedings{ChenT23,
  author       = {Lijie Chen and
                  Roei Tell},
  title        = {When Arthur Has Neither Random Coins Nor Time to Spare: Superfast
                  Derandomization of Proof Systems},
  booktitle    = {Proceedings of the 55th Annual {ACM} Symposium on Theory of Computing,
                  {STOC} 2023, Orlando, FL, USA, June 20-23, 2023},
  pages        = {60--69},
doi          = {10.1145/3564246.3585215},
  publisher    = {{ACM}},
  year         = {2023}
}

@article{Kabanets00,
  author    = {Valentine Kabanets},
  title     = {Easiness Assumptions and Hardness Tests: Trading Time for Zero Error},
  journal   = {J. Comput. Syst. Sci.},
  volume    = {63},
  number    = {2},
  pages     = {236--252},
  year      = {2001},
doi          = {10.1006/JCSS.2001.1763},
}

@article{ShaltielU09,
  author       = {Ronen Shaltiel and
                  Christopher Umans},
  title        = {Low-End Uniform Hardness versus Randomness Tradeoffs for {AM}},
  journal      = {{SIAM} J. Comput.},
  volume       = {39},
  number       = {3},
  pages        = {1006--1037},
  year         = {2009},
  url          = {https://doi.org/10.1137/070698348}
}

@article{GutfreundST03,
  author       = {Dan Gutfreund and
                  Ronen Shaltiel and
                  Amnon Ta{-}Shma},
  title        = {Uniform hardness versus randomness tradeoffs for Arthur-Merlin games},
  journal      = {Comput. Complex.},
  volume       = {12},
  number       = {3-4},
  pages        = {85--130},
  year         = {2003},
  url          = {https://doi.org/10.1007/s00037-003-0178-7},
  doi          = {10.1007/s00037-003-0178-7},
  bibsource    = {dblp computer science bibliography, https://dblp.org}
}

@article{Lu01,
  author       = {Chi{-}Jen Lu},
  title        = {Derandomizing Arthur-Merlin games under uniform assumptions},
  journal      = {Comput. Complex.},
  volume       = {10},
  number       = {3},
  pages        = {247--259},
  year         = {2001},
  url          = {https://doi.org/10.1007/s00037-001-8196-9},
  doi          = {10.1007/s00037-001-8196-9},
  bibsource    = {dblp computer science bibliography, https://dblp.org}
}

@article{ImpagliazzoW01,
  author       = {Russell Impagliazzo and
                  Avi Wigderson},
  title        = {Randomness vs Time: Derandomization under a Uniform Assumption},
  journal      = {J. Comput. Syst. Sci.},
  volume       = {63},
  number       = {4},
  pages        = {672--688},
  year         = {2001},
  url          = {https://doi.org/10.1006/jcss.2001.1780}
}

@BOOK{AB09-book,
  title = {Computational Complexity - {A} Modern Approach},
  publisher = {Cambridge University Press},
  year = {2009},
  author = {Sanjeev Arora and Boaz Barak},
doi = {10.1017/CBO9780511804090},
  isbn = {978-0-521-42426-4},
  urlignore = {http://www.cambridge.org/catalogue/catalogue.asp?isbn=9780521424264}
}

@article{Williams16Derand,
  author    = {R. Ryan Williams},
  title     = {Natural Proofs versus Derandomization},
  journal   = {{SIAM} J. Comput.},
  volume    = {45},
  number    = {2},
  pages     = {497--529},
  year      = {2016},
  url       = {https://doi.org/10.1137/130938219},
  doi       = {10.1137/130938219},
  bibsource = {dblp computer science bibliography, https://dblp.org}
}

@article{TrevisanV07,
  author    = {Luca Trevisan and
               Salil P. Vadhan},
  title     = {Pseudorandomness and Average-Case Complexity Via Uniform Reductions},
  journal   = {Computational Complexity},
  volume    = {16},
  number    = {4},
  pages     = {331--364},
  year      = {2007},
  url       = {https://doi.org/10.1007/s00037-007-0233-x},
  doi       = {10.1007/s00037-007-0233-x},
  bibsource = {dblp computer science bibliography, https://dblp.org}
}

@article{allenderbkmr06,
title = {Power from Random Strings},
author = {Eric Allender and Harry Buhrman and Michal Kouck{\'{y}} and Dieter van Melkebeek and Detlef Ronneburger},
noteignore = {Preliminary version in \href{http://dx.doi.org/10.1109/SFCS.2002.1181992}{FOCS'02}},
journal = {{SIAM} J. Comput.},
number = {6},
volume = {35},
year = {2006},
pages = {1467--1493},
doi = {10.1137/050628994},
}

@article{ImpagliazzoKW02,
  author    = {Russell Impagliazzo and
               Valentine Kabanets and
               Avi Wigderson},
  title     = {In search of an easy witness: exponential time vs. probabilistic polynomial
               time},
  journal   = {J. Comput. Syst. Sci.},
  volume    = {65},
  number    = {4},
  pages     = {672--694},
  year      = {2002},
  url       = {https://doi.org/10.1016/S0022-0000(02)00024-7},
  doi       = {10.1016/S0022-0000(02)00024-7},
  bibsource = {dblp computer science bibliography, https://dblp.org}
}

@inproceedings{Smolensky87,
  author    = {Roman Smolensky},
  title     = {Algebraic Methods in the Theory of Lower Bounds for Boolean Circuit
               Complexity},
  booktitle = {Proceedings of the 19th Annual {ACM} Symposium on Theory of Computing,
               1987},
  pages     = {77--82},
  year      = {1987},
  url       = {https://doi.org/10.1145/28395.28404},
  doi       = {10.1145/28395.28404},
  bibsource = {dblp computer science bibliography, https://dblp.org}
}

@article{BshoutyCGKT96,
  author    = {Nader H. Bshouty and
               Richard Cleve and
               Ricard Gavald{\`{a}} and
               Sampath Kannan and
               Christino Tamon},
  title     = {Oracles and Queries That Are Sufficient for Exact Learning},
  journal   = {J. Comput. Syst. Sci.},
  volume    = {52},
  number    = {3},
  pages     = {421--433},
  year      = {1996},
  url       = {https://doi.org/10.1006/jcss.1996.0032},
  doi       = {10.1006/jcss.1996.0032},
  bibsource = {dblp computer science bibliography, https://dblp.org}
}

@inproceedings{kabanets-cai00,
  author    = {Valentine Kabanets and
               Jin{-}yi Cai},
  title     = {Circuit minimization problem},
  booktitle = {Proceedings of the Thirty-Second Annual {ACM} Symposium on Theory of Computing},
  pages     = {73--79},
  year      = {2000},
  url       = {https://doi.org/10.1145/335305.335314},
  doi       = {10.1145/335305.335314},
  bibsource = {dblp computer science bibliography, https://dblp.org}
}

@inproceedings{McKayMW19,
  author    = {Dylan M. McKay and
               Cody D. Murray and
               R. Ryan Williams},
  title     = {Weak lower bounds on resource-bounded compression imply strong separations
               of complexity classes},
  booktitle = {Proceedings of the 51st Annual {ACM} {SIGACT} Symposium on Theory
               of Computing, {STOC} 2019},
  pages     = {1215--1225},
  publisher = {{ACM}},
  year      = {2019},
  url       = {https://doi.org/10.1145/3313276.3316396},
  doi       = {10.1145/3313276.3316396},
  bibsource = {dblp computer science bibliography, https://dblp.org}
}

@article{Pich15,
  author       = {J{\'{a}}n Pich},
  title        = {Circuit lower bounds in bounded arithmetics},
  journal      = {Ann. Pure Appl. Log.},
  volume       = {166},
  number       = {1},
  pages        = {29--45},
  year         = {2015},
  url          = {https://doi.org/10.1016/j.apal.2014.08.004},
  doi          = {10.1016/j.apal.2014.08.004},
  bibsource    = {dblp computer science bibliography, https://dblp.org}
}

@inproceedings{OliveiraPS18,
  author    = {Igor Carboni Oliveira and
               J{\'{a}}n Pich and
               Rahul Santhanam},
  title     = {Hardness Magnification near State-Of-The-Art Lower Bounds},
  booktitle = {34th Computational Complexity Conference, {CCC} 2019},
  pages     = {27:1--27:29},
  year      = {2019},
  url       = {https://doi.org/10.4230/LIPIcs.CCC.2019.27},
  doi       = {10.4230/LIPIcs.CCC.2019.27},
  bibsource = {dblp computer science bibliography, https://dblp.org}
}

@inproceedings{OliveiraS18,
  author    = {Igor Carboni Oliveira and
               Rahul Santhanam},
  title     = {Hardness Magnification for Natural Problems},
  booktitle = {59th {IEEE} Annual Symposium on Foundations of Computer Science, {FOCS}
               2018},
  pages     = {65--76},
  year      = {2018},
  url       = {https://doi.org/10.1109/FOCS.2018.00016},
  doi       = {10.1109/FOCS.2018.00016},
  bibsource = {dblp computer science bibliography, https://dblp.org}
}

@book{GoldreichBook08,
  author    = {Oded Goldreich},
  title     = {Computational complexity - a conceptual perspective},
  publisher = {Cambridge University Press},
  year      = {2008},
  isbn      = {978-0-521-88473-0},
doi          = {10.1017/CBO9780511804106},
  bibsource = {dblp computer science bibliography, https://dblp.org}
}

@article{Hastad98,
  author    = {Johan H{\aa}stad},
  title     = {The Shrinkage Exponent of {de Morgan} Formulas is 2},
  journal   = {{SIAM} J. Comput.},
  volume    = {27},
  number    = {1},
  pages     = {48--64},
  year      = {1998},
  url       = {https://doi.org/10.1137/S0097539794261556},
  doi       = {10.1137/S0097539794261556},
  bibsource = {dblp computer science bibliography, https://dblp.org}
}

@inproceedings{Tal14,
  author    = {Avishay Tal},
  title     = {Shrinkage of De Morgan Formulae by Spectral Techniques},
  booktitle = {55th {IEEE} Annual Symposium on Foundations of Computer Science, {FOCS}
               2014},
  pages     = {551--560},
  year      = {2014},
  url       = {https://doi.org/10.1109/FOCS.2014.65},
  doi       = {10.1109/FOCS.2014.65},
}

@article{RazborovR97,
  author    = {Alexander A. Razborov and
               Steven Rudich},
  title     = {Natural Proofs},
  journal   = {J. Comput. Syst. Sci.},
  volume    = {55},
  number    = {1},
  pages     = {24--35},
  year      = {1997},
  url       = {https://doi.org/10.1006/jcss.1997.1494},
  doi       = {10.1006/jcss.1997.1494},
  bibsource = {dblp computer science bibliography, https://dblp.org}
}

@inproceedings{cjw19,
  author    = {Lijie Chen and
               Ce Jin and
               R. Ryan Williams},
  title     = {Hardness Magnification for all Sparse {NP} Languages},
  booktitle = {60th {IEEE} Annual Symposium on Foundations of Computer Science, {FOCS}
               2019},
  pages     = {1240--1255},
  year      = {2019},
  url       = {https://doi.org/10.1109/FOCS.2019.00077},
  doi       = {10.1109/FOCS.2019.00077},
}

@inproceedings{gw14,
  author    = {Oded Goldreich and
               Avi Wigderson},
  title     = {On derandomizing algorithms that err extremely rarely},
  booktitle = {Symposium on Theory of Computing, {STOC} 2014},
  pages     = {109--118},
  publisher = {{ACM}},
  year      = {2014},
  url       = {https://doi.org/10.1145/2591796.2591808},
  doi       = {10.1145/2591796.2591808},
  bibsource = {dblp computer science bibliography, https://dblp.org}
}

@inproceedings{LiptonY94,
  author    = {Richard J. Lipton and
               Neal E. Young},
  title     = {Simple strategies for large zero-sum games with applications to complexity
               theory},
  booktitle = {STOC},
  pages     = {734--740},
doi          = {10.1145/195058.195447},
  year      = {1994}
}

@article{AaronsonW09,
  author    = {Scott Aaronson and
               Avi Wigderson},
  title     = {Algebrization: {A} New Barrier in Complexity Theory},
  journal   = {{ACM} Trans. Comput. Theory},
  volume    = {1},
  number    = {1},
  pages     = {2:1--2:54},
  year      = {2009},
  url       = {https://doi.org/10.1145/1490270.1490272},
  doi       = {10.1145/1490270.1490272},
  bibsource = {dblp computer science bibliography, https://dblp.org}
}

@inproceedings{BrodyV10,
  author    = {Joshua Brody and
               Elad Verbin},
  title     = {The Coin Problem and Pseudorandomness for Branching Programs},
  booktitle = {51th Annual {IEEE} Symposium on Foundations of Computer Science, {FOCS}
               2010},
  pages     = {30--39},
  year      = {2010},
  url       = {https://doi.org/10.1109/FOCS.2010.10},
  doi       = {10.1109/FOCS.2010.10},
  bibsource = {dblp computer science bibliography, https://dblp.org}
}

@article{GutfreundST07,
  author    = {Dan Gutfreund and
               Ronen Shaltiel and
               Amnon Ta{-}Shma},
  title     = {If {NP} Languages are Hard on the Worst-Case, Then it is Easy to Find
               Their Hard Instances},
  journal   = {Computational Complexity},
  volume    = {16},
  number    = {4},
  pages     = {412--441},
  year      = {2007},
  url       = {https://doi.org/10.1007/s00037-007-0235-8},
  doi       = {10.1007/s00037-007-0235-8},
  bibsource = {dblp computer science bibliography, https://dblp.org}
}

@inproceedings{ChenHOPRS20,
  author    = {Lijie Chen and
               Shuichi Hirahara and
               Igor Carboni Oliveira and
               J{\'{a}}n Pich and
               Ninad Rajgopal and
               Rahul Santhanam},
  title     = {Beyond Natural Proofs: Hardness Magnification and Locality},
  booktitle = {11th Innovations in Theoretical Computer Science Conference, {ITCS}},
  pages     = {70:1--70:48},
  year      = {2020},
  url       = {https://doi.org/10.4230/LIPIcs.ITCS.2020.70},
  doi       = {10.4230/LIPIcs.ITCS.2020.70},
}

@article{Bar-YossefJKS04,
  author    = {Ziv Bar{-}Yossef and
               T. S. Jayram and
               Ravi Kumar and
               D. Sivakumar},
  title     = {An information statistics approach to data stream and communication
               complexity},
  journal   = {J. Comput. Syst. Sci.},
  volume    = {68},
  number    = {4},
  pages     = {702--732},
  year      = {2004},
  url       = {https://doi.org/10.1016/j.jcss.2003.11.006},
  doi       = {10.1016/j.jcss.2003.11.006},
  bibsource = {dblp computer science bibliography, https://dblp.org}
}

@article{KalyanasundaramS92,
  author    = {Bala Kalyanasundaram and
               Georg Schnitger},
  title     = {The Probabilistic Communication Complexity of Set Intersection},
  journal   = {{SIAM} J. Discrete Math.},
  volume    = {5},
  number    = {4},
  pages     = {545--557},
  year      = {1992},
  url       = {https://doi.org/10.1137/0405044},
  doi       = {10.1137/0405044},
  bibsource = {dblp computer science bibliography, https://dblp.org}
}

@article{Razborov92,
  author    = {Alexander A. Razborov},
  title     = {On the Distributional Complexity of Disjointness},
  journal   = {Theor. Comput. Sci.},
  volume    = {106},
  number    = {2},
  pages     = {385--390},
  year      = {1992},
  url       = {https://doi.org/10.1016/0304-3975(92)90260-M},
  doi       = {10.1016/0304-3975(92)90260-M},
  bibsource = {dblp computer science bibliography, https://dblp.org}
}

@article{BakerGS75,
  author    = {Theodore P. Baker and
  John Gill and
  Robert Solovay},
  title     = {Relativizations of the {$\PTIME =? \NP$} Question},
  journal   = {{SIAM} J. Comput.},
  volume    = {4},
  number    = {4},
  pages     = {431--442},
  year      = {1975},
  url       = {https://doi.org/10.1137/0204037},
  doi       = {10.1137/0204037},
  bibsource = {dblp computer science bibliography, https://dblp.org}
}

@inproceedings{DolevFG13,
  author    = {Shlomi Dolev and
               Nova Fandina and
               Dan Gutfreund},
  title     = {Succinct Permanent Is \emph{NEXP}-Hard with Many Hard Instances},
  booktitle = {Algorithms and Complexity, 8th International Conference, {CIAC} 2013. Proceedings},
  series    = {Lecture Notes in Computer Science},
  volume    = {7878},
  pages     = {183--196},
  publisher = {Springer},
  year      = {2013},
  url       = {https://doi.org/10.1007/978-3-642-38233-8\_16},
  doi       = {10.1007/978-3-642-38233-8\_16},
  bibsource = {dblp computer science bibliography, https://dblp.org}
}

@article{MulmuleyFlip10,
  author    = {Ketan Mulmuley},
  title     = {Explicit Proofs and The Flip},
  journal   = {CoRR},
  volume    = {abs/1009.0246},
  year      = {2010},
  url       = {http://arxiv.org/abs/1009.0246},
  archivePrefix = {arXiv},
  eprint    = {1009.0246},
  bibsource = {dblp computer science bibliography, https://dblp.org}
}

@inproceedings{cjw20,
  author    = {Lijie Chen and
               Ce Jin and
               R. Ryan Williams},
  title     = {Sharp threshold results for computational complexity},
  booktitle = {Proccedings of the 52nd Annual {ACM} {SIGACT} Symposium on Theory of Computing, {STOC} 2020},
  pages     = {1335--1348},
  publisher = {{ACM}},
  year      = {2020},
  doi       = {10.1145/3357713.3384283},
}

@article{vv,
title = "{NP} is as easy as detecting unique solutions",
journal = "Theoretical Computer Science",
volume = "47",
pages = "85--93",
year = "1986",
issn = "0304-3975",
doi = "https://doi.org/10.1016/0304-3975(86)90135-0",
author = "L. G. Valiant and V. V. Vazirani",
}

@inproceedings{Hirahara20,
  author    = {Shuichi Hirahara},
  title     = {Unexpected hardness results for {Kolmogorov} complexity under uniform
               reductions},
  booktitle = {Proccedings of the 52nd Annual {ACM} {SIGACT} Symposium on Theory
               of Computing, {STOC} 2020},
  pages     = {1038--1051},
  publisher = {{ACM}},
  year      = {2020},
  url       = {https://doi.org/10.1145/3357713.3384251},
  doi       = {10.1145/3357713.3384251},
  bibsource = {dblp computer science bibliography, https://dblp.org}
}

@article{MulmuleyVI,
  author    = {Ketan Mulmuley},
  title     = {Geometric Complexity Theory {VI:} the flip via saturated and positive
               integer programming in representation theory and algebraic geometry},
  journal   = {CoRR},
  volume    = {abs/0704.0229},
  year      = {2007},
  url       = {http://arxiv.org/abs/0704.0229},
  archivePrefix = {arXiv},
  eprint    = {0704.0229},
  bibsource = {dblp computer science bibliography, https://dblp.org}
}

@inproceedings{IkenmeyerK20,
  author    = {Christian Ikenmeyer and
               Umangathan Kandasamy},
  title     = {Implementing geometric complexity theory: on the separation of orbit
               closures via symmetries},
  booktitle = {Proccedings of the 52nd Annual {ACM} {SIGACT} Symposium on Theory
               of Computing, {STOC} 2020},
  pages     = {713--726},
  publisher = {{ACM}},
  year      = {2020},
  url       = {https://doi.org/10.1145/3357713.3384257},
  doi       = {10.1145/3357713.3384257},
  bibsource = {dblp computer science bibliography, https://dblp.org}
}

@article{Mulmuley12,
  author    = {Ketan Mulmuley},
  title     = {The {GCT} program toward the \emph{P} vs. \emph{NP} problem},
  journal   = {Commun. {ACM}},
  volume    = {55},
  number    = {6},
  pages     = {98--107},
  year      = {2012},
  url       = {https://doi.org/10.1145/2184319.2184341},
  doi       = {10.1145/2184319.2184341},
  bibsource = {dblp computer science bibliography, https://dblp.org}
}

@article{Nisan91,
  author    = {Noam Nisan},
  title     = {Pseudorandom bits for constant depth circuits},
  journal   = {Comb.},
  volume    = {11},
  number    = {1},
  pages     = {63--70},
  year      = {1991},
  url       = {https://doi.org/10.1007/BF01375474},
  doi       = {10.1007/BF01375474},
  bibsource = {dblp computer science bibliography, https://dblp.org}
}

@inproceedings{FortnowS16,
  title={New non-uniform lower bounds for uniform classes},
  author={Fortnow, Lance and Santhanam, Rahul},
  booktitle={31st Conference on Computational Complexity (CCC 2016)},
  year={2016},
doi          = {10.4230/LIPICS.CCC.2016.19},
}

@inproceedings{Allender01,
  author    = {Eric Allender},
  title     = {When Worlds Collide: Derandomization, Lower Bounds, and {Kolmogorov}
               Complexity},
  booktitle = { Proceedings of the 21st Conference on  Foundations of Software Technology and Theoretical Computer Science ({FST} {TCS} 2001) },
  series    = {Lecture Notes in Computer Science},
  volume    = {2245},
  pages     = {1--15},
  publisher = {Springer},
  year      = {2001},
  url       = {https://doi.org/10.1007/3-540-45294-X\_1},
  doi       = {10.1007/3-540-45294-X\_1},
  bibsource = {dblp computer science bibliography, https://dblp.org}
}

@article{Viola05,
  author    = {Emanuele Viola},
  title     = {The complexity of constructing pseudorandom generators from hard functions},
  journal   = {Comput. Complex.},
  volume    = {13},
  number    = {3-4},
  pages     = {147--188},
  year      = {2005},
  url       = {https://doi.org/10.1007/s00037-004-0187-1},
  doi       = {10.1007/s00037-004-0187-1},
  bibsource = {dblp computer science bibliography, https://dblp.org}
}

@inproceedings{Hastad86,
  author    = {Johan H{\aa}stad},
  title     = {Almost Optimal Lower Bounds for Small Depth Circuits},
  booktitle = {Proceedings of the 18th Annual {ACM} Symposium on Theory of Computing},
  pages     = {6--20},
  publisher = {{ACM}},
  year      = {1986},
  url       = {https://doi.org/10.1145/12130.12132},
  doi       = {10.1145/12130.12132},
  bibsource = {dblp computer science bibliography, https://dblp.org}
}

@article{Ajtai83,
  author    = {Mikl{\'{o}}s Ajtai},
  title     = {Sigma-Formulae on finite structures},
  journal   = {Ann. Pure Appl. Log.},
  volume    = {24},
  number    = {1},
  pages     = {1--48},
  year      = {1983},
  url       = {https://doi.org/10.1016/0168-0072(83)90038-6},
  doi       = {10.1016/0168-0072(83)90038-6},
  bibsource = {dblp computer science bibliography, https://dblp.org}
}

@article{FSS84,
  author    = {Merrick L. Furst and
               James B. Saxe and
               Michael Sipser},
  title     = {Parity, Circuits, and the Polynomial-Time Hierarchy},
  journal   = {Math. Syst. Theory},
  volume    = {17},
  number    = {1},
  pages     = {13--27},
  year      = {1984},
  url       = {https://doi.org/10.1007/BF01744431},
  doi       = {10.1007/BF01744431},
  bibsource = {dblp computer science bibliography, https://dblp.org}
}

@inproceedings{Yao85,
  author    = {Andrew Chi{-}Chih Yao},
  title     = {Separating the Polynomial-Time Hierarchy by Oracles (Preliminary Version)},
  booktitle = {26th Annual Symposium on Foundations of Computer Science,  1985},
  pages     = {1--10},
  year      = {1985},
  url       = {https://doi.org/10.1109/SFCS.1985.49},
  doi       = {10.1109/SFCS.1985.49},
  bibsource = {dblp computer science bibliography, https://dblp.org}
}

@article{MullerPich20,
  author    = {Moritz M{\"{u}}ller and
               J{\'{a}}n Pich},
  title     = {Feasibly constructive proofs of succinct weak circuit lower bounds},
  journal   = {Ann. Pure Appl. Log.},
  volume    = {171},
  number    = {2},
  year      = {2020},
  url       = {https://doi.org/10.1016/j.apal.2019.102735},
  doi       = {10.1016/j.apal.2019.102735},
  bibsource = {dblp computer science bibliography, https://dblp.org}
}

@book{krajivcek2019proof,
  title={Proof complexity},
  author={Kraj{\'\i}{\v{c}}ek, Jan},
  volume={170},
  year={2019},
doi={10.1017/9781108242066},
  publisher={Cambridge University Press}
}

@article{Neciporuk1691966,
  title={On a boolean function},
  author={E. Neciporuk},
  journal={Doklady of the Academy of the USSR},
  volume={169},
  number={4},
  pages={765--766},
  year={1966}
}

@article{razborov1987lower,
  title={Lower bounds on the size of bounded depth circuits over a complete basis with logical addition},
  author={Razborov, Alexander A},
  journal={Mathematical Notes of the Academy of Sciences of the USSR},
  volume={41},
  number={4},
  pages={333--338},
  year={1987},
  publisher={Springer}
}

@inproceedings{PapadimitriouZ83,
  author    = {Christos H. Papadimitriou and
               Stathis Zachos},
  title     = {Two remarks on the power of counting},
  booktitle = {Theoretical Computer Science, 6th GI-Conference, 1983, Proceedings},
  series    = {Lecture Notes in Computer Science},
  volume    = {145},
  pages     = {269--276},
  publisher = {Springer},
  year      = {1983},
  doi       = {10.1007/BFb0009651},
}

@inproceedings{ChenLW20,
  author    = {Lijie Chen and
               Xin Lyu and
               R. Ryan Williams},
  title     = {Almost-Everywhere Circuit Lower Bounds from Non-Trivial Derandomization},
  booktitle = {61st {IEEE} Annual Symposium on Foundations of Computer Science, {FOCS}
               2020},
  pages     = {1--12},
  publisher = {{IEEE}},
  year      = {2020},
  url       = {https://doi.org/10.1109/FOCS46700.2020.00009},
}

@book{mohri2018foundations,
  title={Foundations of machine learning},
  author={Mohri, Mehryar and Rostamizadeh, Afshin and Talwalkar, Ameet},
  year={2018},
  edition = {Second},
  publisher={MIT press}
}

@inproceedings{Atserias06,
  author    = {Albert Atserias},
  title     = {Distinguishing {SAT} from Polynomial-Size Circuits, through Black-Box
               Queries},
  booktitle = {21st Annual {IEEE} Conference on Computational Complexity {(CCC} 2006)},
  pages     = {88--95},
  publisher = {{IEEE} Computer Society},
  year      = {2006},
  url       = {https://doi.org/10.1109/CCC.2006.17}
}

@inproceedings{BogdanovTW10,
  author    = {Andrej Bogdanov and
               Kunal Talwar and
               Andrew Wan},
  title     = {Hard Instances for Satisfiability and Quasi-one-way Functions},
  booktitle = {Innovations in Computer Science - {ICS} 2010, Tsinghua University,
               Beijing, China, January 5-7, 2010. Proceedings},
  pages     = {290--300},
  publisher = {Tsinghua University Press},
  year      = {2010},
  url       = {http://conference.iiis.tsinghua.edu.cn/ICS2010/content/papers/23.html},
  bibsource = {dblp computer science bibliography, https://dblp.org}
}

@article{Jerabek07,
  author    = {Emil Je{\v r}{\'{a}}bek},
  title     = {Approximate counting in bounded arithmetic},
  journal   = {J. Symb. Log.},
  volume    = {72},
  number    = {3},
  pages     = {959--993},
  year      = {2007},
doi          = {10.2178/JSL/1245158087},
}

@inproceedings{Cook75,
  author    = {Stephen A. Cook},
  title     = {Feasibly Constructive Proofs and the Propositional Calculus (Preliminary
               Version)},
  booktitle = {Proceedings of the 7th Annual {ACM} Symposium on Theory of Computing, 1975},
  pages     = {83--97},
  publisher = {{ACM}},
  year      = {1975},
doi          = {10.1145/800116.803756},
}

@inproceedings{pal,
  author    = {Wolfgang Maass},
  title     = {Quadratic Lower Bounds for Deterministic and Nondeterministic One-Tape Turing Machines (Extended Abstract)},
  booktitle = {Proceedings of the 16th Annual {ACM} Symposium on Theory of Computing},
  pages     = {401--408},
  publisher = {{ACM}},
  year      = {1984},
  url       = {https://doi.org/10.1145/800057.808706},
  doi       = {10.1145/800057.808706},
  bibsource = {dblp computer science bibliography, https://dblp.org}
}

@article{Tourlakis01,
  author    = {Iannis Tourlakis},
  title     = {Time-Space Tradeoffs for {SAT} on Nonuniform Machines},
  journal   = {J. Comput. Syst. Sci.},
  volume    = {63},
  number    = {2},
  pages     = {268--287},
  doi          = {10.1006/JCSS.2001.1767},
  year      = {2001}
}

@article{FortnowLMV05,
  author    = {Lance Fortnow and
               Richard J. Lipton and
               Dieter van Melkebeek and
               Anastasios Viglas},
  title     = {Time-space lower bounds for satisfiability},
  journal   = {J. {ACM}},
  volume    = {52},
  number    = {6},
  pages     = {835--865},
doi          = {10.1145/1101821.1101822},
  year      = {2005}
}

@inproceedings{CheraghchiHMY21,
  author       = {Mahdi Cheraghchi and
                  Shuichi Hirahara and
                  Dimitrios Myrisiotis and
                  Yuichi Yoshida},
  title        = {One-Tape Turing Machine and Branching Program Lower Bounds for {MCSP}},
  booktitle    = {38th International Symposium on Theoretical Aspects of Computer Science,
                  {STACS} 2021, March 16-19, 2021, Saarbr{\"{u}}cken, Germany (Virtual
                  Conference)},
  series       = {LIPIcs},
  volume       = {187},
  pages        = {23:1--23:19},
  publisher    = {Schloss Dagstuhl - Leibniz-Zentrum f{\"{u}}r Informatik},
  year         = {2021},
  url          = {https://doi.org/10.4230/LIPIcs.STACS.2021.23},
  doi          = {10.4230/LIPICS.STACS.2021.23},
  bibsource    = {dblp computer science bibliography, https://dblp.org}
}

@inproceedings{Vereshchagin13,
  author       = {Nikolay K. Vereshchagin},
  title        = {Improving on Gutfreund, Shaltiel, and Ta-Shma's Paper "If {NP} Languages
                  Are Hard on the Worst-Case, Then It Is Easy to Find Their Hard Instances"},
  booktitle    = {Computer Science - Theory and Applications - 8th International Computer
                  Science Symposium in Russia, {CSR} 2013, Ekaterinburg, Russia, June
                  25-29, 2013. Proceedings},
  series       = {Lecture Notes in Computer Science},
  volume       = {7913},
  pages        = {203--211},
  publisher    = {Springer},
  year         = {2013},
  url          = {https://doi.org/10.1007/978-3-642-38536-0\_18},
  doi          = {10.1007/978-3-642-38536-0\_18},
  bibsource    = {dblp computer science bibliography, https://dblp.org}
}

@inproceedings{Gutfreund06,
  author       = {Dan Gutfreund},
  title        = {Worst-Case Vs. Algorithmic Average-Case Complexity in the Polynomial-Time
                  Hierarchy},
  booktitle    = {Approximation, Randomization, and Combinatorial Optimization. Algorithms
                  and Techniques, 9th International Workshop on Approximation Algorithms
                  for Combinatorial Optimization Problems, {APPROX} 2006 and 10th International
                  Workshop on Randomization and Computation, {RANDOM} 2006, Barcelona,
                  Spain, August 28-30 2006, Proceedings},
  series       = {Lecture Notes in Computer Science},
  volume       = {4110},
  pages        = {386--397},
  publisher    = {Springer},
  year         = {2006},
  url          = {https://doi.org/10.1007/11830924\_36},
  doi          = {10.1007/11830924\_36},
  bibsource    = {dblp computer science bibliography, https://dblp.org}
}

\end{document}